\documentclass[pdflatex,sn-mathphys]{sn-jnl}

\jyear{2021}%

\theoremstyle{thmstyleone}%
\newtheorem{theorem}{Theorem}
\newtheorem{proposition}[theorem]{Proposition}%

\theoremstyle{thmstyletwo}%
\newtheorem{remark}{Remark}%
\newtheorem*{rep@theorem}{\rep@title}
\newcommand{\newreptheorem}[2]{%
\newenvironment{rep#1}[1]{%
 \def\rep@title{#2 \ref{##1}}%
 \begin{rep@theorem}}%
 {\end{rep@theorem}}}
\makeatother

\newreptheorem{theorem}{Theorem}
\newtheorem{corollary}{Corollary}%
\newreptheorem{corollary}{Corollary}%

\newtheorem{lemma}[theorem]{Lemma}
\newtheorem{definition}{Definition}

\newtheorem{assumption}{Assumption}

\theoremstyle{thmstylethree}%

\begin{document}

\title{Randomized Time Riemannian Manifold Hamiltonian Monte Carlo}


\author*[1]{\fnm{Peter A.} \sur{Whalley}}\email{p.a.whalley@sms.ed.ac.uk}
	
\author[1]{\fnm{Daniel} \sur{Paulin}}\email{dpaulin@ed.ac.uk}
	
\author[1]{\fnm{Benedict} \sur{Leimkuhler}}\email{b.leimkuhler@ed.ac.uk}

\affil*[1]{\orgdiv{Department of Mathematics}, \orgname{, University of Edinburgh.}, \orgaddress{ \city{Edinburgh}, \postcode{EH9 3FD}, \country{UK}}}


\abstract{Hamiltonian Monte Carlo (HMC) algorithms which combine numerical approximation of Hamiltonian dynamics on finite intervals with stochastic refreshment and Metropolis correction are popular sampling schemes, but it is known that they may suffer from slow convergence in the continuous time limit.  A recent paper of Bou-Rabee and Sanz-Serna ({\em Ann. Appl. Prob.}, {\bf 27}:2159-2194, 2017) demonstrated that this issue can be addressed by simply randomizing the duration parameter  of the Hamiltonian paths.  In this article, we use the same idea to enhance the sampling efficiency of a constrained version of HMC, with potential benefits in a variety of application settings. We demonstrate both the conservation of the stationary distribution and the ergodicity of the method. We also compare the performance of various schemes in numerical studies of model problems, including an application to high-dimensional covariance estimation.}




\maketitle

\section{Introduction and Motivation}
Efficient sampling of high dimensional probability distributions is required for Bayesian inference and is a challenge in many fields including biological modelling (\cite{Wi2007}), economic modelling (\cite{Gr2012}), machine learning with large data sets (\cite{Pa2017,Ba2012}) and molecular dynamics (\cite{Pe2015}). A popular approach is Markov chain Monte Carlo, which defines a Markov chain $X_{i+1} \sim p(\cdot \mid X_{i})$ with invariant measure $\mu$ and from which we may estimate expected values from the relation $\mathbb{E}_{X \sim \mu} f(X) \approx \frac{1}{N}\sum^{N}_{i=1}f(X_{i})$; however convergence of such averages  can be slow for high dimensional and multimodal distributions (see e.g. \cite{Qu2018}). Recent attempts to address this problem include the local bouncy particle sampler of \cite{Bouch2018} and the Zig-Zag process of \cite{Bi2019}. These methods can be viewed as piecewise deterministic Markov processes (PDMPs), see \cite{Va2018}. The Randomized Hamiltonian Monte Carlo (RHMC), proposed in \cite{Bo2017} and further studied in \cite{De2021} , evolves a Hamiltonian flow for a duration drawn from an exponential distribution.  In standard HMC the choice of integration time is a challenging task (see \cite{Ho2014}) and mixing can be inefficient for some choices of integration time. By contrast, RHMC does not suffer from this problem as randomization of the duration prevents periodicities. This strategy has been studied from both analytic and numerical perspectives in \cite{Bo2017}. Other recent algorithms have been proposed which build on this idea (for example \cite{Ri2022} and \cite{Kl2022}).

The algorithms discussed above are targeted to sampling from distributions in Euclidean space. The need to work with Riemannian manifolds is motivated by applications where constraints are imposed from modelling considerations or are introduced in order to restrict sampling to a relevant subdomain derived from statistical analysis  (see \cite{Br2012}).   Examples of manifolds include products of spheres or orthogonal matrices which arise in applications in protein configuration modelling with the Fisher-Bingham distribution (\cite{Ha2006}), texture analysis using distributions over rotations (\cite{Ku2004}) and fixed-rank
matrix factorization for collaborative filtering (\cite{Me2011}, \cite{Sa2008}). Methods that sample from probability distributions on manifolds have been considered in \cite{Ha2008}, \cite{Br2012}, \cite{By2013}, \cite{Gi2011}, \cite{Le2012}, \cite{Za2018}, \cite{Di2012}, \cite{Le2019} and \cite{La2021}. In this article, we focus on manifolds defined by algebraic constraints. In order to maintain the constraints, in practice one needs to perform projections at each step of the algorithm, an additional overhead compared to Euclidean MCMC algorithms.

In this paper we propose the Randomized Time Riemannian Manifold Hamiltonian Monte Carlo (RT-RMHMC) method, an RHMC scheme for Riemannian manifolds.  We establish invariance under a compactness assumption of the desired measure in the (small stepsize limit) continuous-time PDMP version of our method, where the algorithm is rejection free. Further, we demonstrate the invariance of the discretized method with Metropolis-Hastings adjustment and prove ergodicity of the discretized method with Metropolis-Hastings adjustment.   We show in numerical experiments that this method has improved robustness, demonstrating for example that the convergence rate is relatively flat in the choice of mean time parameter; these results mirror those obtained for the Euclidean version of the method.   Moreover, we compare RT-RMHMC to a constrained underdamped Langevin integrator g-BAOAB introduced in \cite{Le2016}.  

To our knowledge, there is no theoretical or numerical treatment of RHMC in the manifold setting and there has been no theoretical treatment of Riemannian Hamiltonian Monte Carlo methods in the continuous time setting. We provide a first result to estabilish invariance of a continuous time Riemannian Hamiltonian Monte Carlo method in the compact setting.   A biased RHMC method was recently introduced (see \cite{Kl2022}) which has event rates which depend on the position in the state space, these state dependent event rates can be incorporated into our RHMC Riemannian framework when the framework is unadjusted. We note that in the appendix of that article,  a version of RHMC is introduced in the setting of adapting the metric for sampling on Euclidean space but not for working on a Riemannian manifold.   

The remainder of this article is organised as follows.  In the next section we describe the algorithm and provide invariance in the continuous time setting under a compactness assumption.  Section 3 considers the numerical implementation with and without Metropolis test.  Section 4 provides conservation of the stationary distribution of the discretized algorithm and the ergodicity of the method with Metropolis-Hastings adjustment. Section 5 discusses numerical experiments and Section 6 gives some thoughts on future developments.   We include several appendices addressing the generator,  the invariance of the target measure and the irreducibility of the scheme, from which ergodicity necessarily follows.

%

\section{Algorithm}
Let $(\mathcal{M},g)$ be a $d$-dimensional Riemannian manifold and $T\mathcal{M}$ denote its tangent bundle. Let $G(x)$ denote the positive definite matrix associated to the metric $g$ at $x \in \mathcal{M}$. 
Consider a target distribution on $\mathcal{M}$ with density
\[
\pi_{\mathcal{H}}(x) = \frac{1}{Z_{\mathcal{M}}}\exp{(- U_{\mathcal{H}}(x))},
\]
with respect to $\sigma_{\mathcal{M}}(dx)$, the surface measure (Hausdorff measure) of $\mathcal{M}$ defined by $\sigma_{\mathcal{M}}(dx) = \sqrt{\det{G(x)}}dx$ and $Z_{\mathcal{M}} = \int_{\mathcal{M}}\exp{(- U_{\mathcal{H}}(x))}\sigma_{\mathcal{M}}(dx)$, which we assume to be finite. Consider an extension of the distribution to $T\mathcal{M}$ as
\begin{equation}\label{measure}
\mu(dz) = \frac{1}{Z_{T\mathcal{M}}} \exp{(-H(x,v))}\lambda_{T \mathcal{M}}(dz),
\end{equation}
where $\lambda_{T\mathcal{M}}(dz)$ is the Liouville measure of $T\mathcal{M}$, $H$ is defined by \begin{equation} \label{hamiltonian}
H(x,v) = U_{\mathcal{H}}(x) + \frac{1}{2}v^{T}G(x)^{-1}v = U(x) + \frac{1}{2}\log{\{(2\pi)^{d}\det{G(x)}\}} + \frac{1}{2}v^{T}G(x)^{-1}v
\end{equation}
for $(x,v) \in T\mathcal{M}$ and $Z_{T\mathcal{M}} = \int_{T\mathcal{M}}\exp{(-H(x,v))}\lambda_{T \mathcal{M}}(dz),$ which is finite when $Z_{\mathcal{M}}$ is. We have that 
\[\mu(dz) = \pi(x)\sigma_{\mathcal{M}}(dx)\psi(x)(dv), \]
where $\psi(x)(dv)$ is simply the Gaussian measure on $T_{x}\mathcal{M}$ given by 
\[\psi(x)(dv)  =  \frac{1}{\sqrt{\{(2\pi)^{d}\det{G(x)}\}}}\exp{\{-\frac{1}{2}v^{T}G(x)^{-1}v\}} \sigma_{T_{x}\mathcal{M}}(dv) \]
in local coordinates and $\sigma_{T_{x}\mathcal{M}}(dv)$ is the Lebesgue measure on $T_{x}\mathcal{M}$. In particular we have that $\mu$ has marginal distribution $\pi_{\mathcal{H}}$ with respect to the Hausdorff measure (\cite{Gi2011},\cite{By2013},\cite{Le2010}[Section 3.3.2]).

We will define a stochastic process which is a Riemannian version of the Randomized Hamiltonian Monte Carlo of \cite{Bo2017}. The stochastic process follows constrained Hamiltonian dynamics for an time duration $t$ sampled from $t \sim \exp{(\lambda)}$ for some rate $\lambda > 0$ before an event. This event is a random velocity refreshment under the distribution $\psi(x)$.

Algorithm \ref{alg:1} defines Randomized time Riemannian Manifold Hamiltonian Monte Carlo (RT-RMHMC) with rate parameter $\lambda > 0$, and Hamiltonian dynamics governed by the Hamiltonian $H(x,v) = U_{\mathcal{H}}(x) + \frac{1}{2}v^{T}G(x)^{-1}v$ defined on $T\mathcal{M}$. This stochastic process has invariant measure $\mu(z) = \exp{(-H(z))}$ with respect to the Liouville measure on $T\mathcal{M}$.

\begin{algorithm}[H]
\footnotesize
\SetAlgoLined
	\begin{itemize}
		\item Initialise $x_{0}$ arbitrarily on $\mathcal{M}$ and sample $v_{0} \sim \psi(x_{0})$ on $T_{x_{0}}\mathcal{M}$ such that $(x_{0},v_{0}) \in T\mathcal{M}$.
		\item Initialise $t_{0} = 0$.
		\item for $k = 1,2,...$ do
		\begin{itemize}
			\item Update time via $t_{k} = t_{k-1} + \delta t $, where $\delta t \sim \exp{(\lambda)}$.
			\item Evolve over $[t_{k-1},t_{k}]$ Hamilton's equations
			with initial condition $(x(t_{k-1}),v(t_{k-1})) = (x_{t_{k-1}},v_{t_{k-1}})$.
			\item Set $(x_{s},v_{s}) = (x(s),v(s)) \text{ for } s \in [t_{k-1},t_{k}).$
			\item Set $x_{t_k} = x(t_k)$ and sample $v_{t_{k}} \sim \psi(x_{t_k})$ such that $(x_{t_k},v_{t_k}) \in T\mathcal{M}.$
		\end{itemize}
	\end{itemize}
	\caption{RT-RMHMC}
	\label{alg:1}
\end{algorithm}
To sample from a distribution $\pi$ with respect to the Hausdorff measure we define $U = - \log{\pi}$ under the assumption that $\pi$ is integrable on $\mathcal{M}$.
 
We can define the generator for this stochastic process as
\begin{equation}\label{generator}
    \mathcal{L}f(z) = X_{H}(f(z)) + \lambda(Qf(z) - f(z)),
\end{equation}
where 
\[Qf(x,v) := \frac{1}{\sqrt{\{(2\pi)^{d}\det{G}(x)\}}}\int_{T_{x}\mathcal{M}} \exp{\{-\frac{1}{2}\xi^{T}G(x)^{-1}\xi\}}f(x,\xi) d\xi \]
is the transition kernel for a completely randomized velocity refreshment according to a Gaussian distribution on the tangent space $T_{x}\mathcal{M}$ and $X_{H}$ is the Hamiltonian vector field associated to $H$. In the Supplementary Material, we will prove that this is the generator of this stochastic process in Section \ref{PDMPs-and-their-invariant-measures} and invariance of the measure in Section \ref{invariant-measure} under a compactness assumption. Our main theoretical result about Algorithm \ref{alg:1} is the following.

\begin{repcorollary}{cor:invmeasureRTRMHMC}[Invariant measure for RT-RMHMC]
Let $(P_{t})_{t \geq 0}$ be the transition semigroup of a simulation of Algorithm \ref{alg:1} with characteristics $(\varphi,\lambda,Q)$ on $T\mathcal{M}$ and Hamiltonian $H \in C^{2}(T\mathcal{M})$, where $(\mathcal{M},g)$ is a compact smooth Riemannian manifold and $\varphi$ is the Hamiltonian flow associated to the Hamiltonian. Let $\mu$ be the measure on $(T\mathcal{M}, \mathcal{B}(T\mathcal{M}))$ given by
\[ \mu(dz) \propto e^{-H(x,v)}d\lambda_{T \mathcal{M}}(z),\]
where $d\lambda_{T\mathcal{M}}$ is the Liouville measure of $T\mathcal{M}$.
Then $\mu$ is invariant for RT-RMHMC.
\end{repcorollary}


\section{Constrained Symplectic Integrator and Metropolis-Hastings adjustment}

In this section, we will state some more broadly implementable versions of Algorithm \ref{alg:1} that are applicable when the Hamiltonian dynamics cannot be solved exactly. We start with a brief introduction to Lagrangian and Hamiltonian dynamics with constraints based on \cite{Lee2017}[Chapter 3] .

Consider manifolds $\mathcal{M}$ embedded in $\mathbb{R}^{d}$ that can be described by algebraic equations
\[
\mathcal{M}:=\{x \in \mathbb{R}^{d} \mid g_{i}(x) = 0, i = 1,...,m \} \subset \mathbb{R}^{d}
\]
where $g_{i}: \mathbb{R}^{d} \to \mathbb{R}$ $i= 1,...,m$ are continuously differentiable functions with linearly independent gradient functions for all $x \in \mathcal{M}$. 

We refer to  such a submanifold as an algebraic constraint manifold. We can express the Euler-Lagrange equations as an orthogonal projection of the Euler-Lagrange equations in $\mathbb{R}^{d}$ onto the constraint manifold, hence we have
\[
\frac{d}{dt}\left(\frac{\partial L(x,\dot{x})}{\partial \dot{x}}\right) - \frac{\partial L(x,\dot{x})}{\partial x} + \sum^{m}_{i=1}\lambda_{i}\frac{\partial g_{i}(x)}{\partial x} = 0, 
\]
where $\lambda_{i}$ are Lagrange multipliers for each of the constraints. We can then define an augmented Lagrangian function $L^{a}:T^{*}M \times \mathbb{R}^{m} \to \mathbb{R}$  by $L^{a}(x,\dot{x},\lambda) = L(x,\dot{x}) + \sum^{m}_{i=1} \lambda_{i}g_{i}(x)$. Then the Euler-Lagrange equations can be expressed as 
\[
\frac{d}{dt}\left(\frac{\partial L^{a}(x,\dot{x},\lambda)}{\partial \dot{x}}\right) - \frac{\partial L^{a}(x,\dot{x},\lambda)}{\partial x} = 0
\]
and the augmented Hamiltonian function $H^{a}:T^{*}\mathcal{M} \times \mathbb{R}^{m} \to \mathbb{R}$ as $H^{a}(x,\mu,\lambda) = \mu \cdot \dot{x} - L^{a}(x,\dot{x},\lambda)$, and we therefore obtain Hamilton's equations (see \cite{Ha2007})
\[
	\dot{x} = \frac{\partial H^{a}(x,\mu,\lambda)}{\partial \mu} \qquad 
	\dot{\mu} = -\frac{\partial H^{a}(x,\mu,\lambda)}{\partial x}.
\]
We next introduce a new formulation of RT-RMHMC for constraint manifolds which we will use for numerical simulation. Note that a constraint manifold Hamiltonian Monte Carlo method was introduced in \cite{Br2012}, but with a deterministic duration parameter. We will use the same notation as that used in \cite{Br2012} to introduce randomized time into this algorithm. 

Let us denote our constraints $c(x) := (g_{i}(x),...,g_{m}(x))^{T}$ and let $C(x) = \frac{\partial c}{\partial x}$ denote the Jacobian of the constraints, which we assume to have full rank everywhere. Define a Hamiltonian of the constrained system as $H(x,v) = U_{\mathcal{H}}(x) + K(v)$, where $K(v) = \frac{1}{2}v^{T}G(x)^{-1}v$ is the kinetic energy and $v$ lies in the cotangent space,  $\mathcal{T}^{*}_{x}\mathcal{M} = \{ v \mid C(x) \frac{\partial H}{\partial v}(x,v) = 0\}$. The dynamics of the constrained system in terms of the Hamiltonian is thus given by
\[
	\dot{v} = -\frac{\partial H}{\partial x} + C(x)^{T}\lambda,\qquad
	\dot{x} = \frac{\partial H}{\partial v},\qquad \text{such that }
	c(x) = 0,
\]
where we remark that we can naturally identify the tangent and cotangent spaces and bundles.

If we let $\pi_{\mathcal{H}}$ be our target measure with respect to the Hausdorff measure. We also let $U_{\mathcal{H}}(x) =  -\log \pi_{\mathcal{H}}(x)$ be the potential energy of our constrained system. We can simulate the constrained Hamiltonian dynamics using \cite{An1983}. However, if we know $\pi_{\mathcal{H}}$ explicitly we can avoid computation of the metric tensor by assuming our system is isometrically embedded in Euclidean space. Under this assumption we can then consider Algorithm \ref{alg:2}, which is an explicit algorithm for simulation of Randomized time constrained Hamiltonian Monte Carlo (RT-CHMC). We will discuss and justify the embedding assumption further in section \ref{embedded manifolds}.

\begin{algorithm}[H]
    \footnotesize
    \SetAlgoLined
	\begin{itemize}
		\item Initialise $x_{0}$ arbitrarily on $\mathcal{M}$ and sample $v_{0} \sim \mathcal{N}(0,I \mid C(x_{0})v_{0} = 0)$.
		\item Initialise $t_{0} = 0$.
		\item for $k = 1,2,...$ do
		\begin{itemize}
			\item Update time via $t_{k} = t_{k-1} + \delta t $, where $\delta t \sim \exp{(1/\lambda)}$.
			\item Evolve over $[t_{k-1},t_{k}]$ Hamilton's equations subject to constraints $c(\cdot)$
			with initial condition $(x(t_{k-1}),v(t_{k-1})) = (x_{t_{k-1}},v_{t_{k-1}})$ and Hamiltonian $H(x,v) = -\log{\pi_{\mathcal{H}}(x)} + \frac{1}{2}v^{T}v$.
			\item Set $(x_{s},v_{s}) = (x(s),v(s)) \text{ for } s \in [t_{k-1},t_{k})$
			\item Set $x_{t_k} = x(t_k)$ and sample $v_{t_{k}} \sim \mathcal{N}(0,I \mid C(x_{t_k})v_{t_k} = 0)$ 
		\end{itemize}
	\end{itemize}
	\caption{RT-CHMC}
	\label{alg:2}
\end{algorithm}

In Algorithm \ref{alg:2}  we sample the Gaussian distribution on the tangent space at a point on $\mathcal{M}$. We can do this by sampling a Gaussian distributed vector and then projecting this orthogonally. To orthogonally project a momentum vector onto $T^{*}\mathcal{M}$ and correctly resample the momentum in Algorithm \ref{alg:2} at $x \in \mathcal{M}$ we apply the projector \[P_{\mathcal{M}}(x) := I - C(x)^T (C(x)C(x)^T)^{-1} C(x).\]
\begin{proposition}
	If $v' \sim \mathcal{N}(0,I)$ then $v = P_{\mathcal{M}}(x)v'$ is distributed according to $v \sim \mathcal{N}(0,I \mid C(x)v = 0)$
\end{proposition}
\begin{proof}
	Can be found in \cite{Gr2021}.
\end{proof}


\subsection{Embedded Manifolds}\label{embedded manifolds}
We next introduce the theory of manifold embeddings as it was presented in \cite{By2013} to show that numerical simulation of RT-CHMC is in fact simulation of RT-RMHMC on constraint manifolds.

If we know the form of the distribution $\pi_{\mathcal{H}}$ with respect to the Hausdorff measure, then we can avoid the computation of the metric tensor and the lack of a global coordinate system (\cite{By2013}). We achieve this using isometric embeddings, remarking that every Riemannian manifold can be isometrically embedded in Euclidean space due to the Nash embedding theorem (\cite{Na1956}). If we have an isometric embedding $\xi : \mathcal{M} \to \mathbb{R}^{n}$, then considering a path $q(t)$ on $\mathcal{M}$, the path $x(t) = \xi(q(t))$ is such that $\dot{x}_{i}(t) = \sum_{j} \frac{\partial x_{i}}{\partial q_{j}} \dot{q}_{j}(t)$. The phase space $(q,p)$, where $\dot{q} = G^{-1}p,$ can then be transformed to the embedded phase space $(x,v),$ where
\[
v = \dot{x} = XG(q)^{-1}p = X(X^{T}X)^{-1}p, \text{ where } X_{ij} = \frac{\partial x_{i}}{\partial q_{j}}, 
\]
since $G = X^{T}X$ due to the fact that the embedding is isometric and preserves inner products (see \cite{By2013}). Now the Hamiltonian (eq \ref{hamiltonian}) is
\[
H(x,v) = -\log{ \pi_{\mathcal{H}}(x)} + \frac{1}{2}v^{T}v 
\]
in terms of coordinates $(x,v).$
When considering sampling of the velocities in Algorithm \ref{alg:1} and Algorithm \ref{alg:2}, since $p \sim \mathcal{N}(0,G(q))$, we have
\[
v \sim \mathcal{N}(0,X(X^{T}X)^{-1}X^{T}),
\]
where $X(X^{T}X)^{-1}X^{T}$ is the orthogonal projection onto the tangent space of the embedded manifold (\cite{By2013}). Therefore we can sample from $\mathcal{N}(0,I)$ and project onto the tangent space to obtain a necessary sample. The Hamiltonian is thus expressed in a form which is independent of the  metric (provided we know the density with respect to the Hausdorff measure). We now introduce the numerical integrator's (RATTLE) scheme \cite{Le2004}[Chapter 7]:
\begin{align*}
	x_{n+1} &= x_{n} + \Delta t v_{n+1/2}\\
	v_{n+1/2} &= v_{n} -\frac{\Delta t}{2}\nabla_{x}U(x_{n}) - \frac{\Delta t}{2}C(x_{n})^{T} \lambda^{n}_{(r)} \qquad \qquad \quad \hspace{2.2mm} \text{such that } c(x_{n+1}) = 0\\
	v_{n+1} &= v_{n+1/2} -\frac{\Delta t}{2}\nabla_{x}U(x_{n+1}) - \frac{\Delta t}{2} C(x_{n+1})^{T} \lambda^{n+1}_{(v)}\qquad  \text{ such that } C(x_{n+1})v_{n+1} = 0,
\end{align*}
where we solve for $\lambda^{n}_{(r)}$ and $\lambda^{n+1}_{(v)}$ at each iteration so that the iterates lie in the tangent bundle. We solve for $\lambda^{n}_{(r)}$ (a non-linear system of equations) by cycling through the constraints, adjusting one multiplier at each iteration.  Denote by $C_{i}$ the $i$th row of $C$ and we first initialise
\[
Q := \overline{x}_{n+1} = x_{n} + \Delta t v_{n} - \frac{\Delta t^{2}}{2}  \nabla_{x}U(x_{n}).
\]
Next we cycle through the list of constraints one after another as follows: for each $i = 1,...,m$ compute
\[
\Delta \Lambda_{i} := \frac{g_{i}(Q)}{C_{i}(Q)C_{i}(x_{n})},
\]
and update $Q$ by $Q := Q - C_{i}(x_{n})^{T}\Delta \Lambda_{i}$ until $g_{i}(Q)<tol$ for all $i =1,...,m$, where $tol$ is a certain prescribed tolerance. Then we set $x_{n+1} = Q$ and have $x_{n+1} \in \mathcal{M}$ within the tolerance.  (Note that other stopping criteria could be used (see \cite{Or2000}).) We solve for $\lambda^{n+1}_{(v)}$ by solving the linear system:
\[
\left ( C(x_{n})C(x_{n})^{T} \right ) \lambda^{n}_{(v)}  = C(x_{n})\left(\frac{2}{\Delta t} v_{n-1/2} - \nabla_{x}U(x_{n}) \right) .
\]
Once the linear system has been solved we obtain $(x_{n+1},v_{n+1}) \in T^{*}\mathcal{M}$.
\begin{theorem} 
	Let $\mathcal{M}$ be a constraint manifold. Let $H \in  C^{2}(T\mathcal{M})$, the RATTLE numerical integrator of the Hamiltonian system defined by $H$ in $T\mathcal{M}$  is symmetric, symplectic and of order 2. Further it respects the manifold constraints.
\end{theorem}
\begin{proof}
	Given in \cite{Le1994}.
\end{proof}

\subsection{Metropolis Hastings Adjustment}\label{sec:metropolis_adjustment}

Let $\Psi^{L}_{\Delta t}: T\mathcal{M} \to T\mathcal{M}$ be the numerical integrator defined by $L$ steps of RATTLE with stepsize $\Delta t$. This integrator approximates the Hamiltonian dynamics. For theoretical purposes we will also define the map $N: T\mathcal{M} \to T \mathcal{M}$ which negates the momentum term i.e. $N(x,v) \equiv  (x,-v)$. Note that this leaves the Hamiltonian invariant and due to the fact that the momentum is resampled this has no affect on the samples from $\pi_{\mathcal{H}}$. We will define the following Metropolised RT-RMHMC, where we sample $T \sim \exp{(\lambda)}$ and fix a maximum time length $\Delta t_{\max{}}$ below the stability threshold of the numerical integrator. Then we choose the number of leapfrog steps $L$ to be $\lceil T/\Delta t_{\max{}} \rceil$. Having chosen $L$ in this way, we set $\Delta t = T/L \leq \Delta t_{\max}$. At each step we perform $L$ RATTLE steps with stepsize $\Delta t$. We propose this method of discretisation instead of purely randomising the stepsize and fixing a number of leapfrog steps to avoid numerical instabilities in the numerical integrator. One could also propose fixing a stepsize within the numerical stability threshold of the integrator and simply sampling an integer number of leapfrog steps geometrically to randomise the time. However our proposed method closer relates to the continuous dynamics without the issues due to numerical instabilities.

\begin{remark}\label{reversibility}
	For large choices of stepsize $\Delta t$ it has been shown that $\Phi^{L}_{\Delta t}$ is not reversible where RATTLE is used to integrate on the manifold, see \cite{Le2019,Za2018}. In \cite{Le2019} they propose to combat this by adding a reversibility check incorporated into the Metropolis-Hastings adjustment, although in practice such checks may be neglected in favor of an implicit assumption that $\Delta t$ is sufficiently small to avoid non-reversibility issues. We will investigate this further in section \ref{numerical-results}. 
\end{remark}

In light of Remark \ref{reversibility}, we include $\textnormal{Rev}(\cdot)$ as a additional (optional) accept-reject condition which implements a reversibility check (following \cite{Le2019}). In numerical experiments we examine the stepsize threshold where the reversibility condition fails (See Fig. \ref{fig:reversibility-check}).

\begin{algorithm}[H]
    \footnotesize
\SetAlgoLined
	\begin{itemize}
		\item Initialise $\Delta t_{\max}$ within the stability threshold.
		\item Initialise $x_{0}$ arbitrarily on $\mathcal{M}$ and sample $v_{0} \sim \mathcal{N}(0,I \mid C(x_{0})v_{0} = 0)$.
		\item Initialise $t_{0} = 0$.
		\item for $k = 1,2,...$ do
		\begin{enumerate}
			\item Sample $v_{k-1} \sim \mathcal{N}(0,I \mid C(x_{k-1})v_{k-1} = 0)$.
			\item  \begin{itemize}
				\item Sample $T \sim \exp{(\lambda)}$.
				\item Set $L = \lceil T/\Delta t_{\max} \rceil$ and $\Delta t = T/L$.
				\item Set $(x^{*},v^{*}) = \Psi(x_{k-1},v_{k-1}) = \textnormal{Rev}(N( \Psi^{L}_{\Delta t}(x_{k-1},v_{k-1})))$.
				\item Accept $(x^{*},v^{*})$ with probability $\min{\{ 1, \exp{\{ H(x,v)-H(x^{*},v^{*})\}}\}}$ and set $(x_{k},v_{k}) = (x^{*},v^{*})$.
				\item Otherwise $(x_{k},v_{k}) = (x_{k-1},v_{k-1})$.
			\end{itemize}
		\end{enumerate}
	\end{itemize}
	\caption{RT-CHMC with Metropolis-Hastings step}
	\label{alg:3}
\end{algorithm}

\begin{remark}
Our framework can be adapted to handle inequality constraints by incorporating an additional rejection condition in the Metropolis-Hastings step, which rejects samples which aren't within the boundary.

This will be used in our application in Section \ref{sec:high-dim-cov} to impose a half-normal prior on some dimensions of our Bayesian model.
\end{remark}

\section{Ergodicity}

We will now prove ergodicity and exact invariance of the desired measure of the discrete time algorithm with metropolis-hastings adjustment. We will provide ergodicity under two assumptions by the same technique as \cite{Br2012} and restating some of their results.

\begin{proposition} \label{mu-invariant}
	Assuming that $\Psi$ is reversible for $\Delta t_{max} > 0$ then $\mu$ is invariant with respect to the Markov kernel proposed in Algorithm \ref{alg:3}.
\end{proposition}

\begin{proof}
See Section \ref{sec:proofMetropolis} of the Supplementary Material.
\end{proof}

\begin{assumption} \label{discrete-assumption-1}
Let $\mathcal{M} \in \{x \in \mathbb{R}^{n} \mid c(x) = 0 \}$ be Riemannian manifold which is connected, smooth and differentiable. We assume that $\nicefrac{\partial c}{\partial x}$ is full rank everywhere.
\end{assumption}

\begin{assumption} \label{discrete-assumption-2}
 Let $\mathcal{M}$ be a Riemannian manifold which satisfies assumption \ref{discrete-assumption-1}. For $x \in \mathcal{M}$ we define $\mathcal{B}_{r}(x) = \{x' \in \mathcal{M} \mid d(x',x) \leq r \}$ to be the geodesic ball of radius $r$ of $x$. We assume that there exists a $r > 0$ such that for every $x \in \mathcal{M}$ and $x' \in \mathcal{B}_{r}(x)$ there exists a unique choice of Lagrange multipliers and velocity $v \in T_{x}\mathcal{M}$, $v' \in T_{x'}\mathcal{M}$ for which $(v',x') = \Psi^{L}_{\Delta t}(v,x)$ for sufficently small $\Delta t$.
\end{assumption}

\begin{theorem}[Accessibility] \label{accessibility}
Let $U \in C^{2}(\mathcal{M})$, and assuming assumption \ref{discrete-assumption-1}. For any $x_{0}, x_{1} \in \mathcal{M}$ and $\Delta t$ sufficently small, there exists finite $v_{0} \in T \mathcal{M}$, $v_{1} \in T \mathcal{M}$ and Lagrange multipliers $\lambda_0$, $\lambda_1$ such that $(v_1 , x_1 ) = \Psi_{\Delta t}(v_{0},x_{0}).$ 
\end{theorem}

\begin{proof}
Found in \cite{Br2012}[Theorem 2] and is an extension of the results of \cite{Ma2001}[Theorem 2.1.1] and \cite{Ha2006a}[Theorem 5.6, Section IX.5.2].
\end{proof}

\begin{theorem}[$\mu$-irreducible]\label{irreducible}
Let $U \in C^{2}(\mathcal{M})$, and under assumptions \ref{discrete-assumption-1} and \ref{discrete-assumption-2} we have that for any $x \in \mathcal{M}$, and measurable set $A \subset \mathcal{M}$ with positive measure. Then there exists an $n \in \mathbb{N}$ such that
\[K^{n}(x,A) > 0, \]
where $K$ denotes the marginal transition kernel defined on $\mathcal{M}$ of Algorithm \ref{alg:3}.
\end{theorem}
\begin{proof}
See Section \ref{sec:proofMetropolis} of the Supplementary Material.
\end{proof}

\begin{lemma}[Aperiodic]\label{Aperiodic} Let $U \in C^{2}(\mathcal{M})$ and under Assumptions \ref{discrete-assumption-1} and \ref{discrete-assumption-2} Algorithm \ref{alg:3} is aperiodic.
\end{lemma}
\begin{proof}
Proof given in \cite{Br2012}[Lemma 1].
\end{proof}
\begin{theorem}[Ergodicity]
 Let $U \in C^{2}(\mathcal{M})$ and under Assumptions \ref{discrete-assumption-1} and \ref{discrete-assumption-2} we have for $\mu-$almost all starting values $x$
 \[\lim_{t \to \infty} \int_{\mathcal{M}}\lvert K^{t}(x,y) - \pi_{\mathcal{H}}(y)\vert \sigma_{\mathcal{M}}(dy) = 0. \]
\end{theorem}
\begin{proof}
Since Algorithm \ref{alg:3} is $\mu-$invariant by Theorem \ref{mu-invariant}, $\mu-$irreducible by Theorem \ref{irreducible} and aperiodic by Theorem \ref{Aperiodic}, the required result holds by \cite{Ti1994}[Theorem 1].
\end{proof}


\section{Numerical results}\label{numerical-results}
We perform numerical simulations of the RT-RMHMC algorithm and compare to the RMHMC algorithm of \cite{Br2012,Gi2011}, specifically exploring the underlying dynamics of the two processes. MCMC schemes are used to approximate expected values of certain functions $f$ over some distribution with pdf $\pi$
\[ \mathbb{E}_{\pi}(f) = \int f(x) \pi(x)dx, \]
where we can estimate this quantity using our MCMC scheme by
\[ 
\overline{f} := \mathbb{E}_{\pi}(f) \approx \frac{1}{M} \sum^{M}_{i=1} f(X^{i}), 
\]
where $X^{i}$ are the Markov chain from our MCMC method.  We  quantify the convergence rate associated to approximation of $\mathbb{E}_{\pi}(f)$ by considering the integrated autocorrelation function and essential sample size.

\subsection{g-BAOAB}
As a comparison method we implemented the g-BAOAB integrator of \cite{Le2016},  a numerical integrator for constrained underdamped Langevin dynamics. Constrained underdamped Langevin dynamics can be described by
\begin{align*}
	\dot{x} &= v \\
	\dot{v} &= -\nabla_{x} U(x) - \gamma v + \sqrt{2\gamma}R(t) - C(x)^{T}\lambda,\\
	\textnormal{such that } 0 &= c(x)\textnormal{ and } 0=  C(x)v,
\end{align*}
where $\gamma$ is a friction coefficient and $R(t)$, is a vector-valued, stationary, zero-mean Gaussian process. The numerical integrator g-BAOAB is a splitting method for such dynamics, which uses similar constrained integrators as that of RT-RMHMC. We note that g-BAOAB is a biased sampling algorithm due to the error in the numerical integrator. For a full description of g-BAOAB and a discussion of the sampling error we refer  to \cite{Le2016}.

\subsection{Test Examples}

We next provide examples of distributions on  implicitly defined manifolds embedded in Euclidean space, with the distributions defined with respect to the Hausdorff measure of the manifold. We will consider two types of constraint manifolds: spheres and Stiefel manifolds.

\subsubsection*{Bingham-Von Mises-Fisher distribution on $S^{n}$}

The first test case is the Bingham-Von Mises-Fisher (BVMF) distribution defined on the $n-$dimensional sphere embedded in $\mathbb{R}^{n+1}$, that is $S^{n}:= \{x \in \mathbb{R}^{n+1} \mid \sum^{n+1}_{i=1}x^{2}_{i} = 1 \}$. The BVMF distribution is the exponential family on $S^{n} \subset \mathbb{R}^{n+1}$ with density of the form
\[
\pi_{\mathcal{H}}(x) \propto \exp{\{ c^{T}x + x^{T}Ax\}},
\]
where $c \in \mathbb{R}^{n+1}$ and $A \in M_{n+1}(\mathbb{R})$ is a symmetric matrix.

We compare the integrated autocorrelation (IAC) of $-\log \pi_{\mathcal{H}}$ of the RT-RMHMC method to that of the RMHMC method introduced in \cite{Gi2011} for a number of distributions with parameters defined in the captions. We also compare the maximum IAC of $x_{i}$ for $i = 1,...,n$ to compare the worst efficiency of the mixing in all dimensions. We compare methods by setting the event rate parameter $\lambda$ of RT-RMHMC to be the deterministic duration parameter of RMHMC (running the dynamics for this duration before momentum randomization). We then compute the integrated autocorrelation of  $-\log \pi_{\mathcal{H}}$ and $x_{i}$ for $i = 1,...,n$ for the two methods for varying choices of $\lambda$ by a Monte Carlo averaging procedure as described in section \ref{IAC}. Regarding the reversibility issue for large choices of stepsize (as discussed in Section \ref{sec:metropolis_adjustment}), for the geometries and distributions chosen, this is shown to exhibit behaviour as in Figure \ref{fig:reversibility-check}, where there is a dramatic change in reversibility failure for a small change in step-size. Before this point all samples generated satisfy reversibility conditions. We simply chose stepsizes which are below this threshold in our simulations.

The results are presented in Figure \ref{IAC-2}. We choose the stepsize in RATTLE to be $\Delta t = 0.001$ and sample $N = 1,000,000$ events  with a burn time of $10\%$ of samples before we compute the Monte Carlo average. We also use lags of up to $M = N/50$, $2$ percent of the number of samples used to estimate  the IAC. As our choice of $\Delta t$ is small, the acceptance rate is high so this process is close to the continous version. The IAC compares the efficiency of the continuous processes.

\begin{figure}
	\centering
	\includegraphics[width=0.4\textwidth]{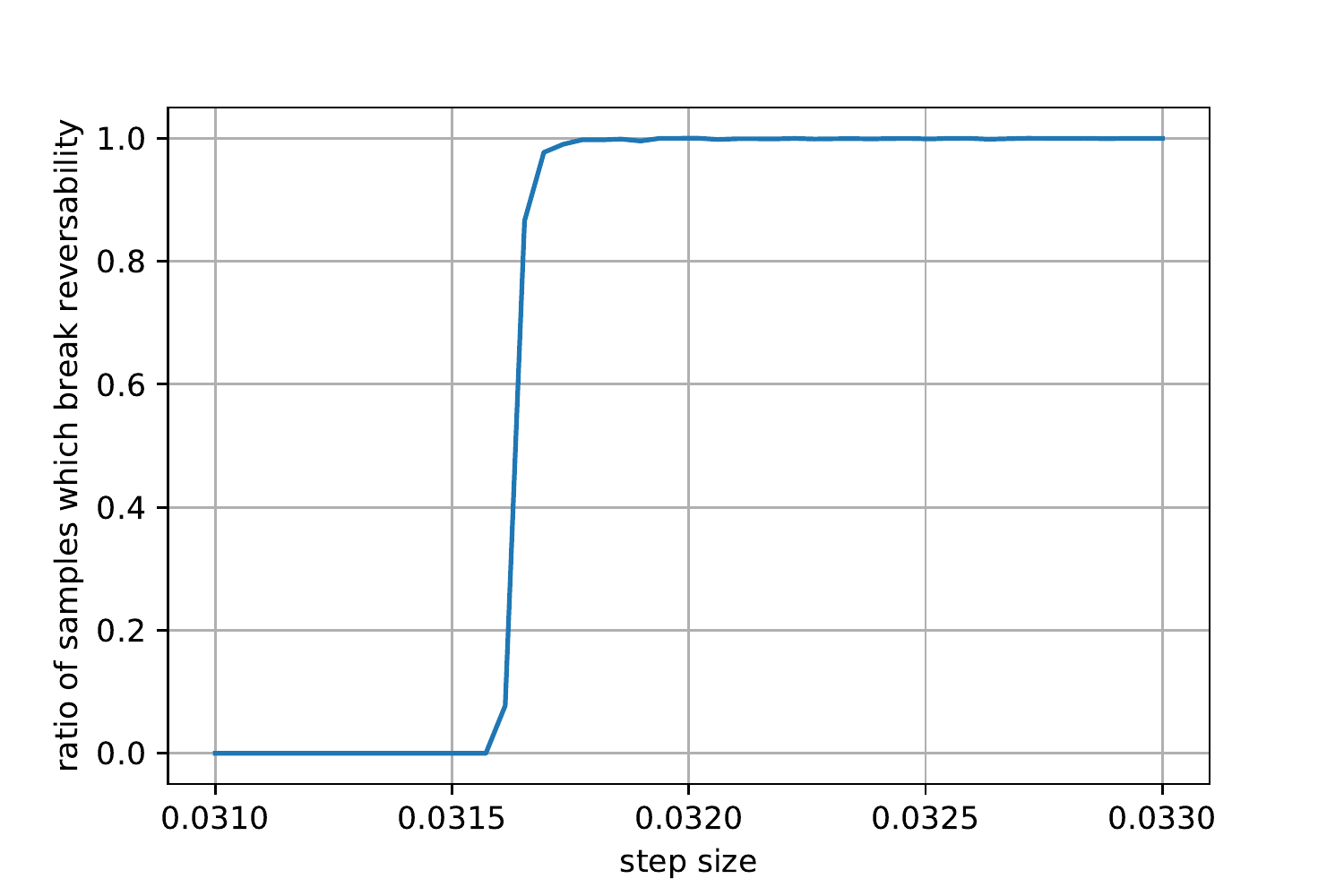}
	\caption{Ratio of samples out of $10^6$ samples which don't satisfy reversibility check for different choices of $\Delta t$ for the BVMF distribution with parameters $A = \text{diag}(-1000,0,1000)$ and $c = (100,0,0)$.}
	\label{fig:reversibility-check}
\end{figure}


\begin{figure}
	\centering
	\begin{subfigure}[b]{0.45\textwidth}
		\centering
		\includegraphics[width = \textwidth]{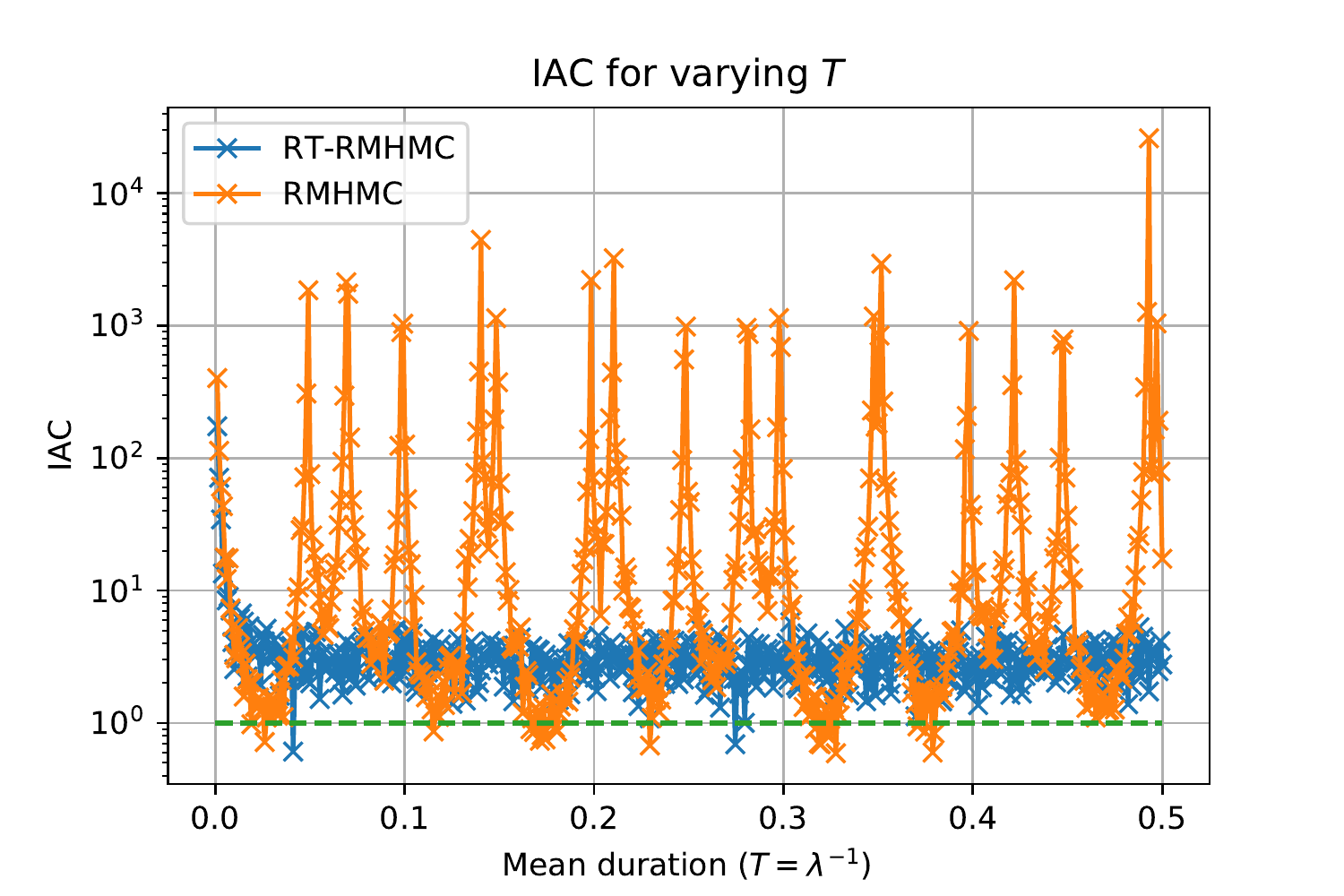}
		\caption{}
		
	\end{subfigure}
	\hfill
	\begin{subfigure}[b]{0.45\textwidth}
		\centering
		\includegraphics[width=\textwidth]{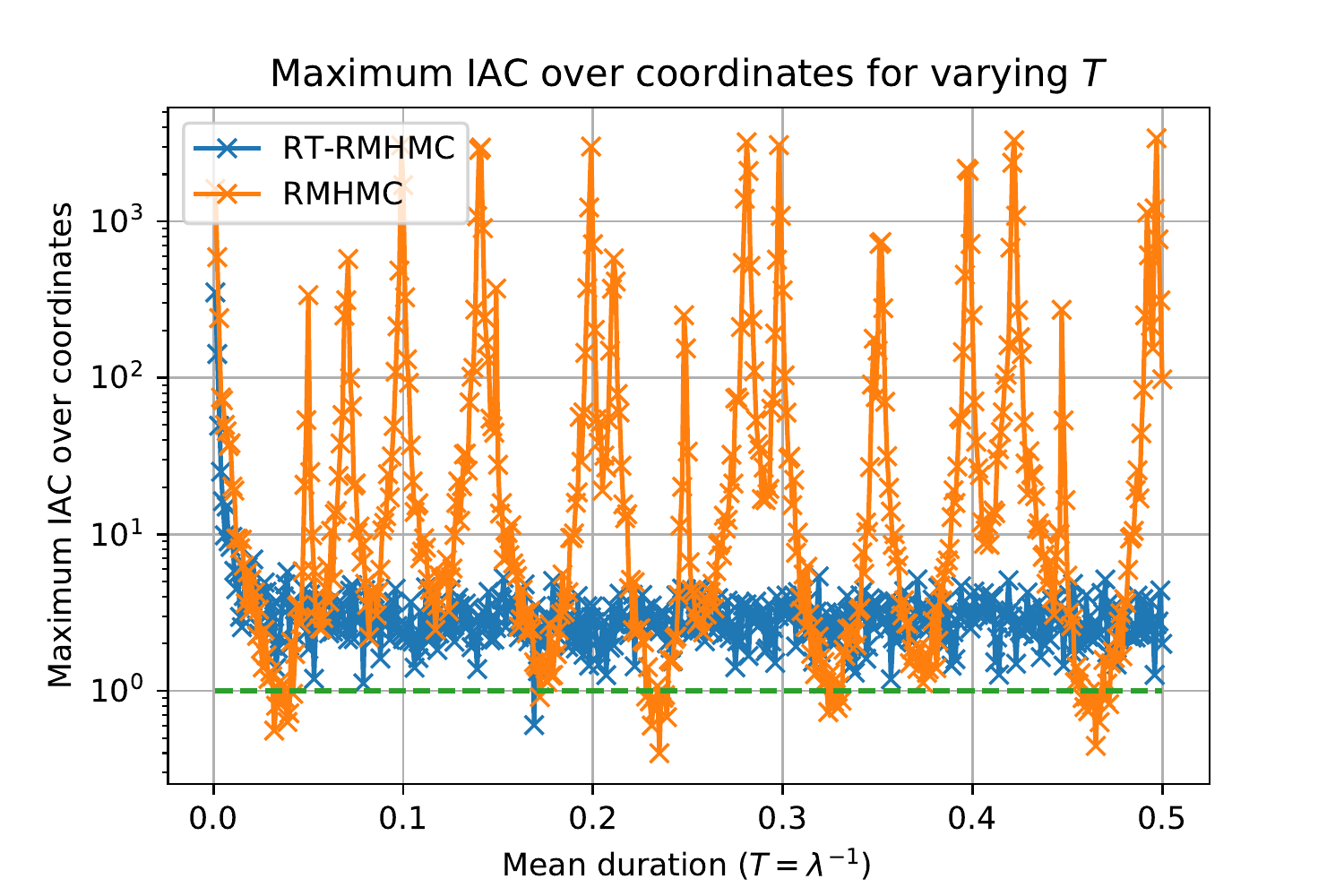}
		\caption{}
		
	\end{subfigure}
	
	\caption{IAC estimates for different choices of $\lambda$ for the BVMF distribution with parameters $A = \text{diag}(-1000,0,1000)$ and $c = (100,0,0)$. Left: IAC of $-\log \pi_{\mathcal{H}}$. Right: Maximum IAC over $x_{1},x_{2}$ and $x_{3}$.}
	\label{IAC-2}
\end{figure}

\begin{figure}
	\centering
	\begin{subfigure}[b]{0.4\textwidth}
		\centering
		\includegraphics[width=\textwidth]{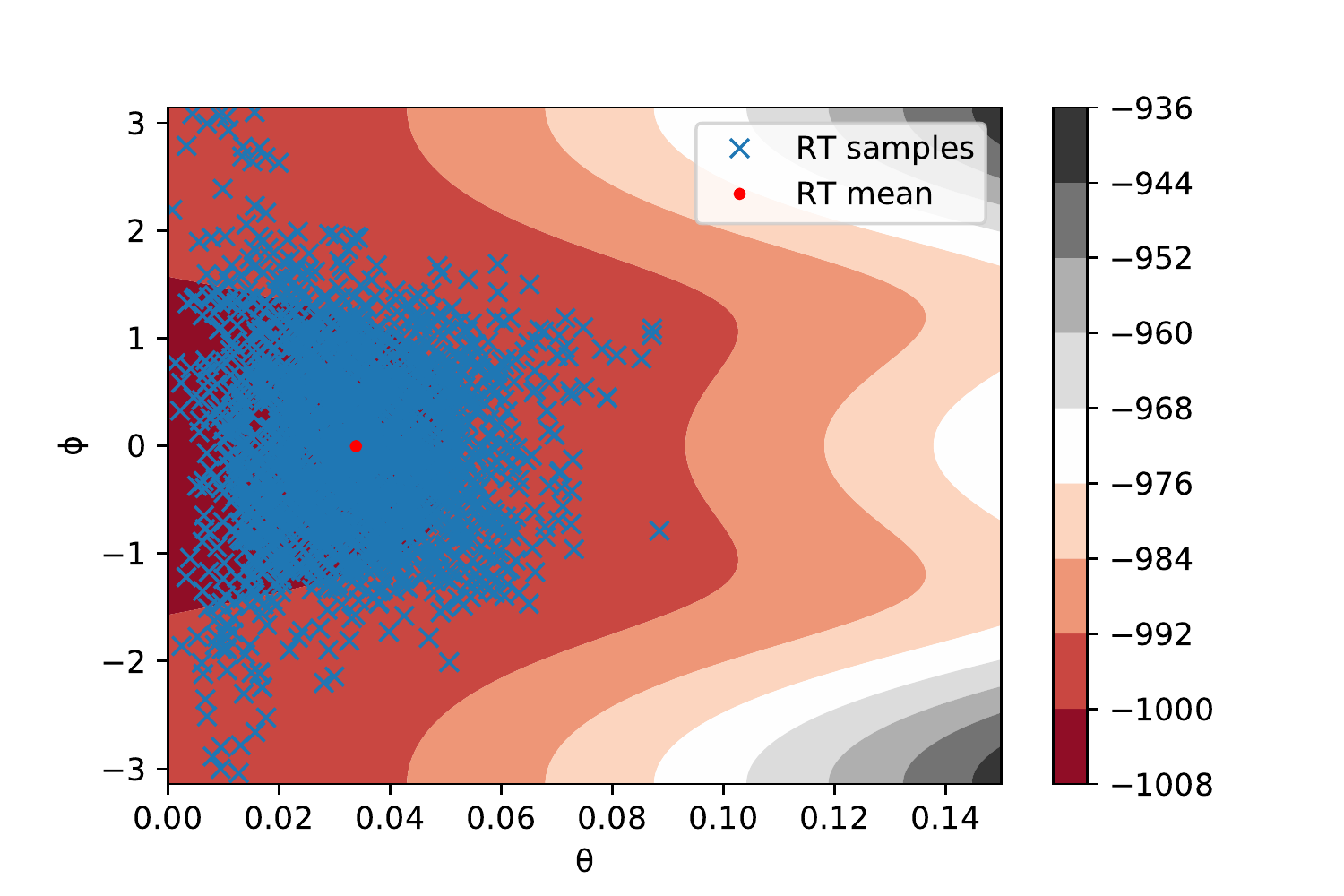}
		\caption{}
		
	\end{subfigure}
	\hspace{1.5cm}
	\begin{subfigure}[b]{0.4\textwidth}
		\centering
		\includegraphics[width=\textwidth]{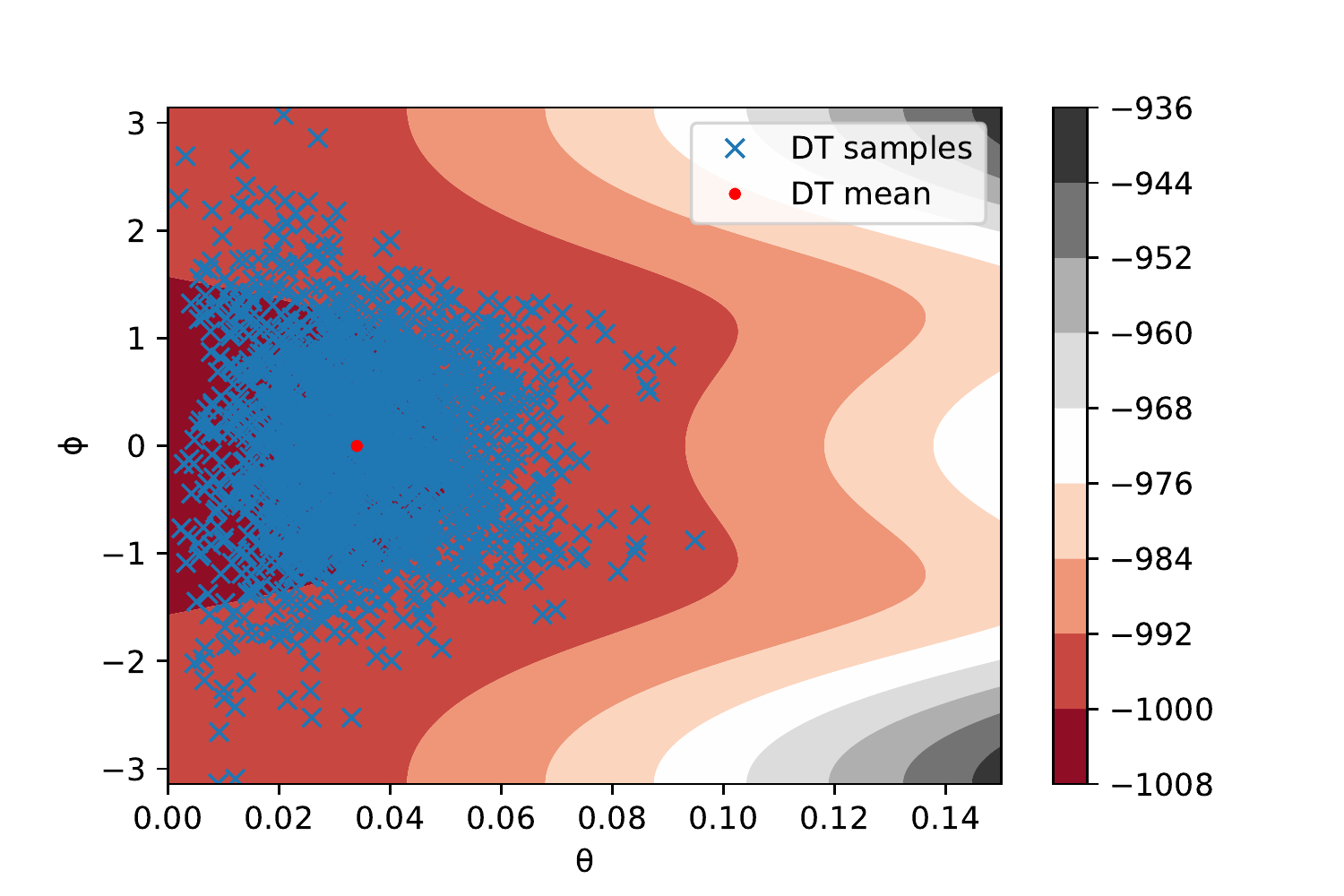}
		\caption{}
		
	\end{subfigure}
	\begin{subfigure}[b]{0.4\textwidth}
		\centering
		\includegraphics[width=\textwidth]{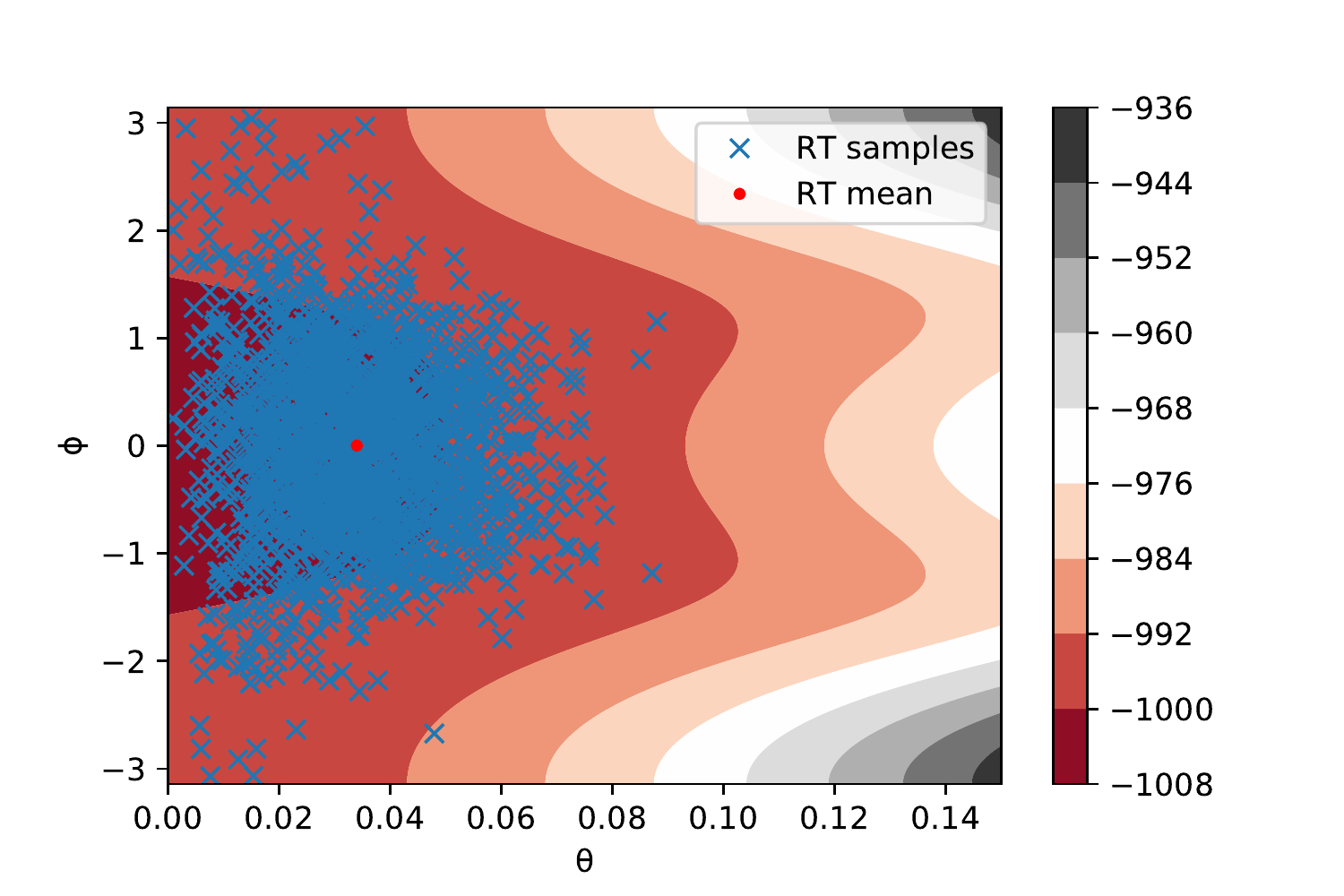}
		\caption{}
		
	\end{subfigure}
	\hspace{1.5cm}
	\begin{subfigure}[b]{0.4\textwidth}
		\centering
		\includegraphics[width=\textwidth]{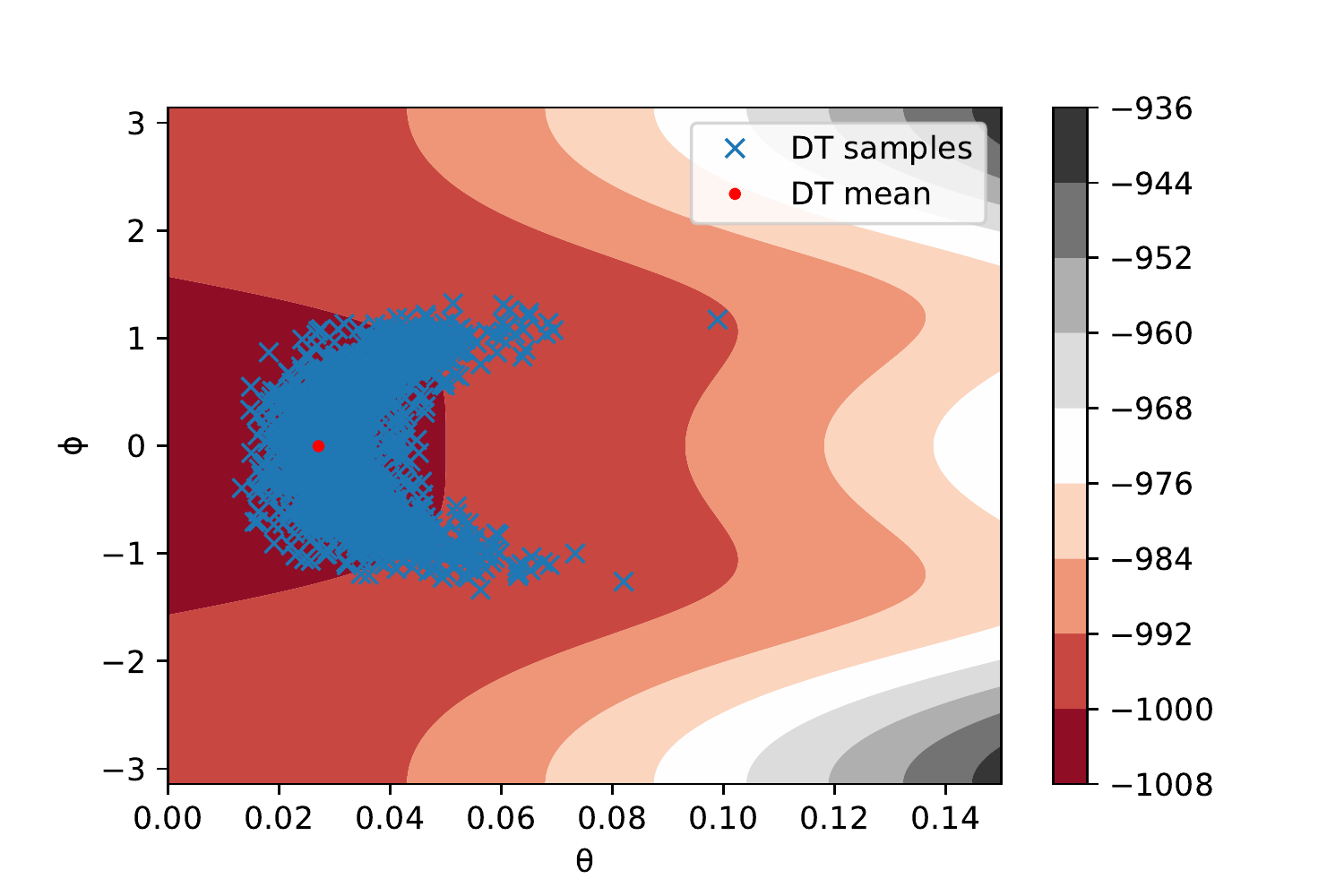}
		\caption{}
	\end{subfigure}
	\caption{Contour plot of $-\log \pi_{\mathcal{H}}$ for the BVMF distribution with parameters $A = \text{diag}(-1000,0,1000)$ and $c = (100,0,0)$. The axis being a 2D parameterisation of $S^2$. The points are 2000 samples after a 8000 sample burn in time. Upper left: RT-RMHMC for $\lambda^{-1} = 0.09$. Upper right: RMHMC for $\lambda^{-1} = 0.09$. Lower left: RT-RMHMC for $\lambda^{-1} = 0.1$. Lower right: RMHMC for $\lambda^{-1} = 0.1$.}
	\label{fig:Contour}
\end{figure}

\begin{figure}[H]
	\centering
	\begin{subfigure}{0.4\textwidth}
		
		\includegraphics[width=\linewidth]{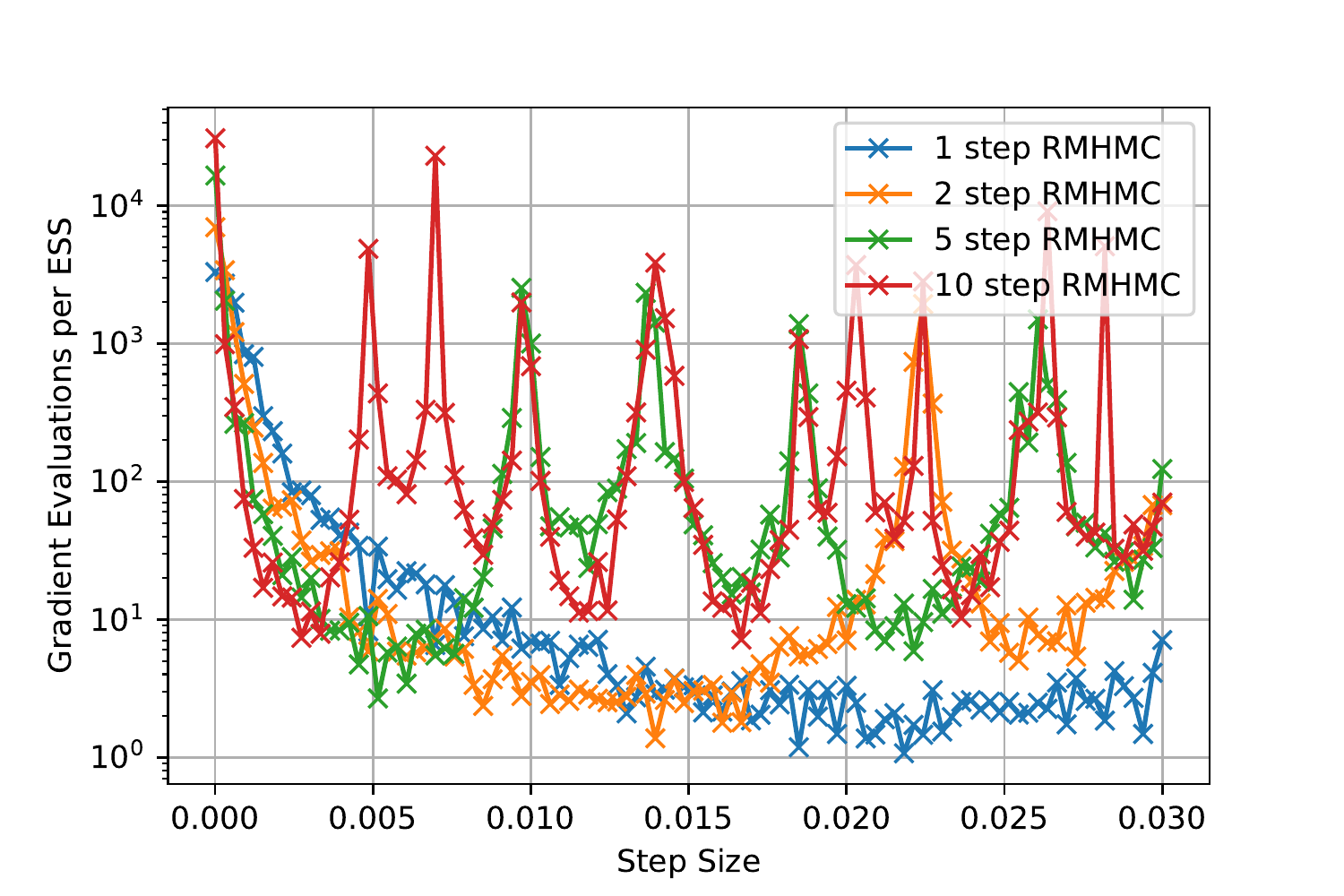}         \caption{}
	\end{subfigure}
	\hfill
	\begin{subfigure}{0.4\textwidth}
		
		\includegraphics[width=\linewidth]{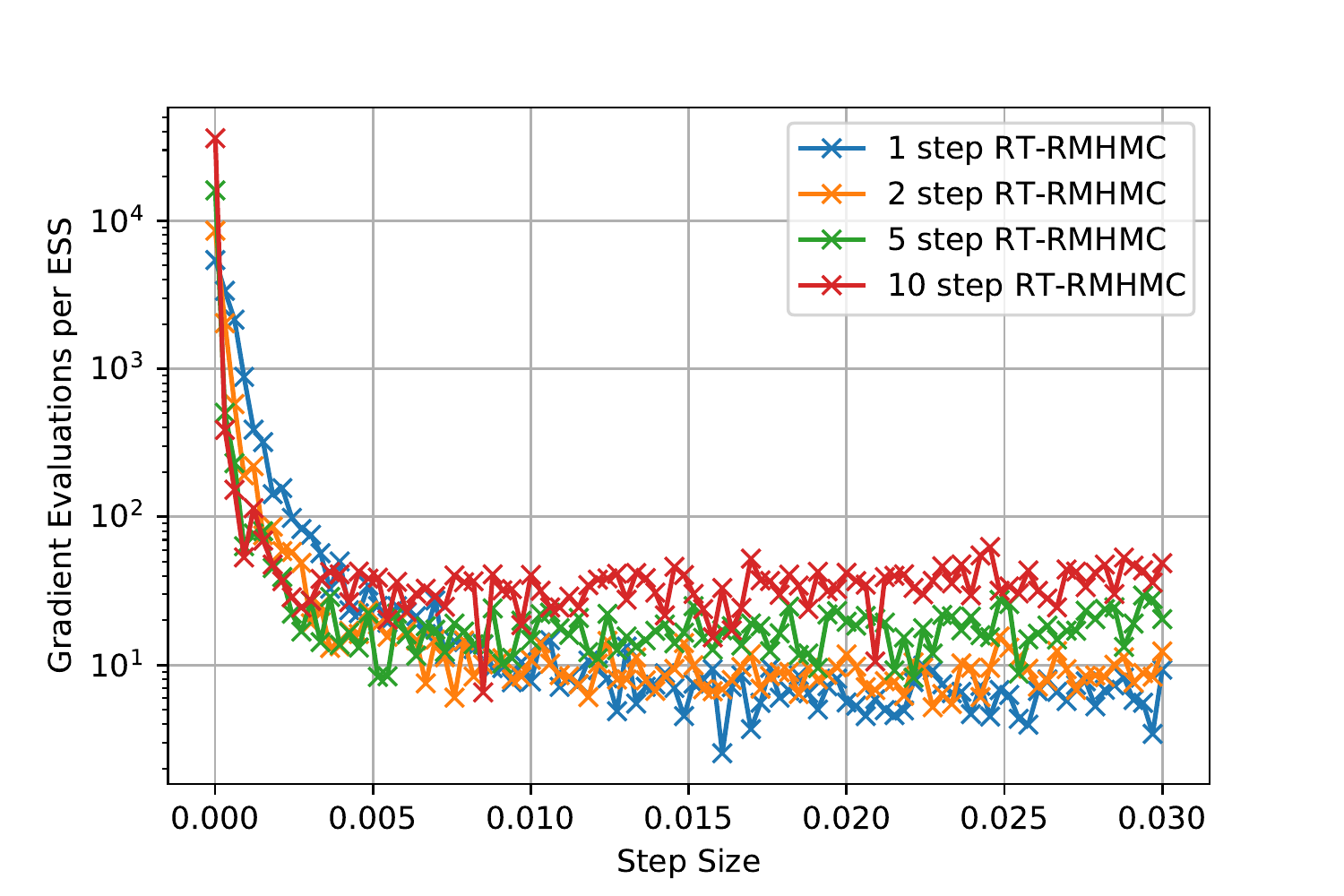}         \caption{}
	\end{subfigure}
	\hfill
	\begin{subfigure}{0.4\textwidth}
		
		\includegraphics[width=\linewidth]{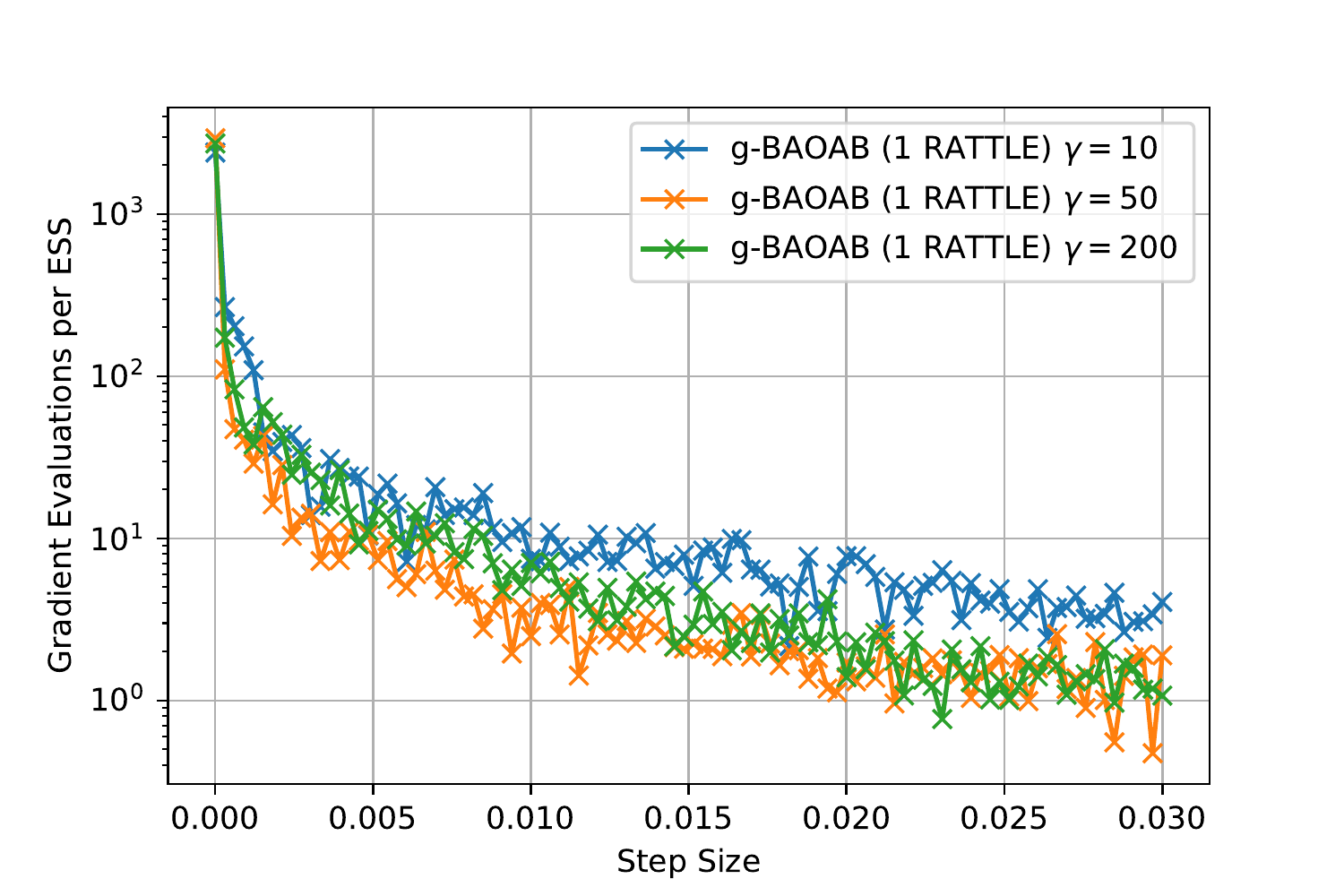}
		\caption{}
		
	\end{subfigure}
	\hfill
	\begin{subfigure}{0.4\textwidth}
		
		\includegraphics[width=\linewidth]{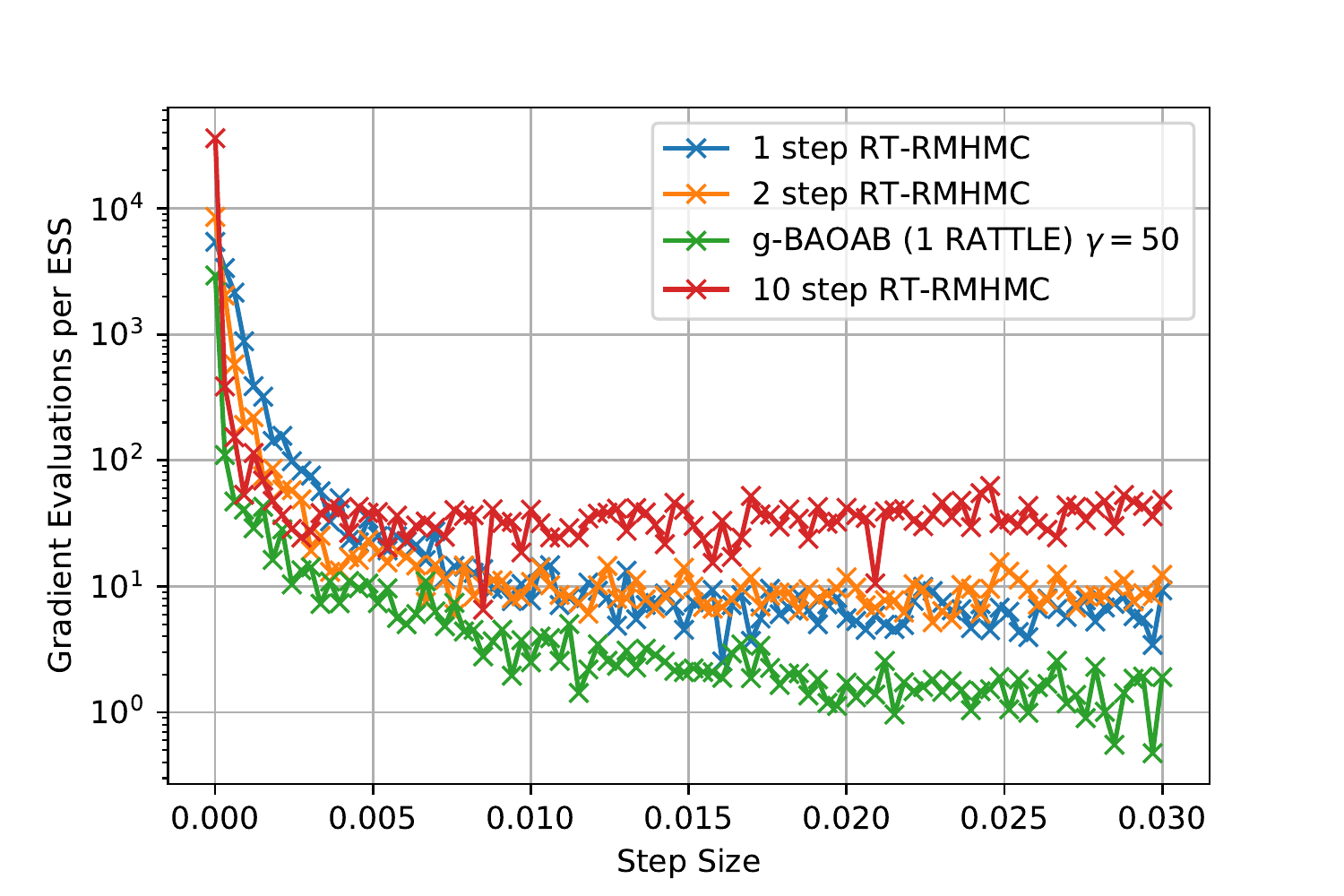}
		\caption{}
		
	\end{subfigure}
	\caption{Gradient evaluation per ESS estimates of $-\log \pi_{\mathcal{H}}$ for the BVMF distribution using 100,000 samples with parameters $A = \text{diag}(-1000,0,1000)$ and $c = (100,0,0)$ and for varying choices of step-size. Upper left: RMHMC. Upper right: RT-RMHMC. Lower left: g-BAOAB. Lower right: g-BAOAB and RT-RMHMC.}
	\label{fig:ESS-BAOAB}
	
\end{figure}

\begin{figure}[H]
	\centering
	\begin{subfigure}{0.4\textwidth}
		
		\includegraphics[width=\linewidth]{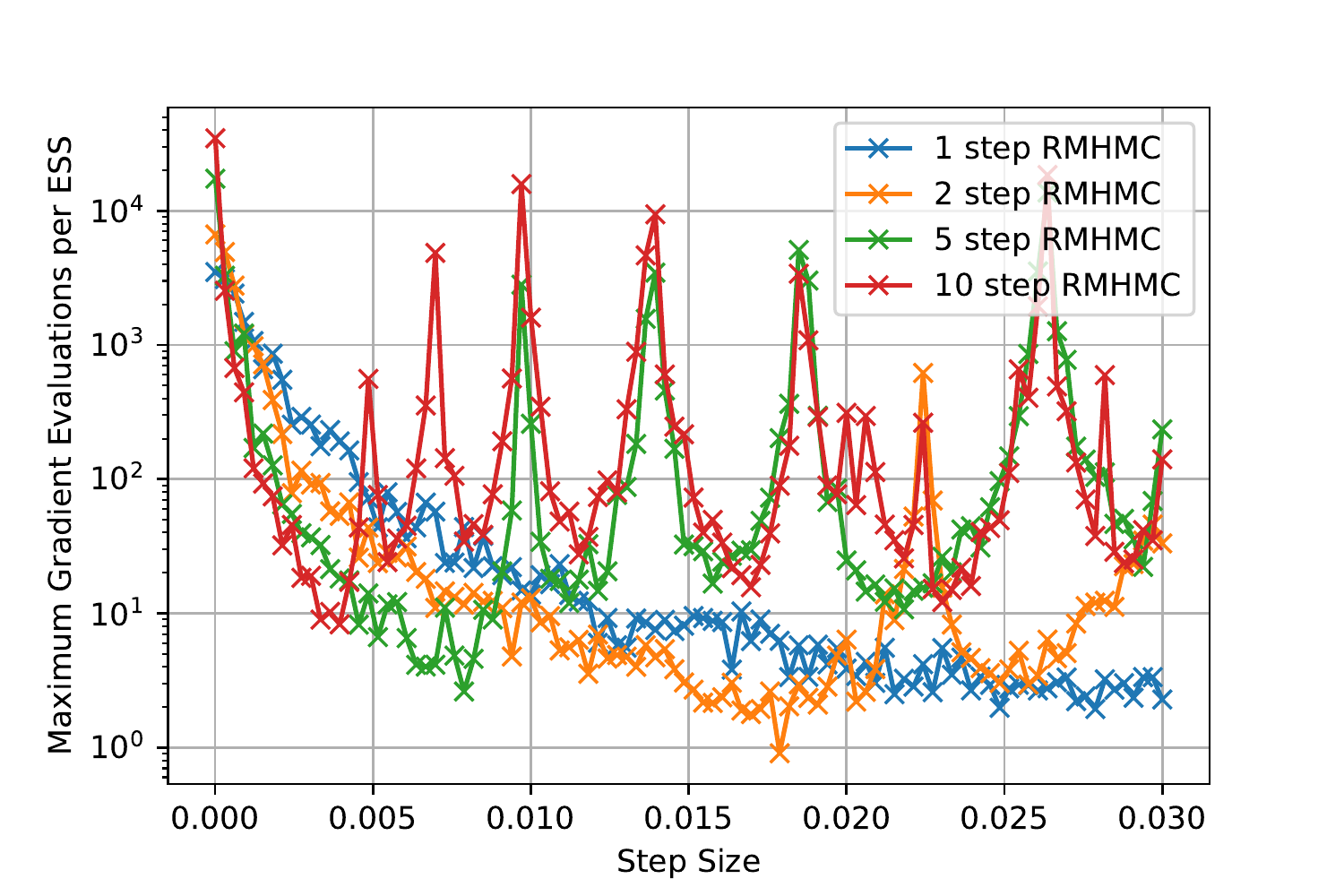}         \caption{}
	\end{subfigure}
	\hfill
	\begin{subfigure}{0.4\textwidth}
		
		\includegraphics[width=\linewidth]{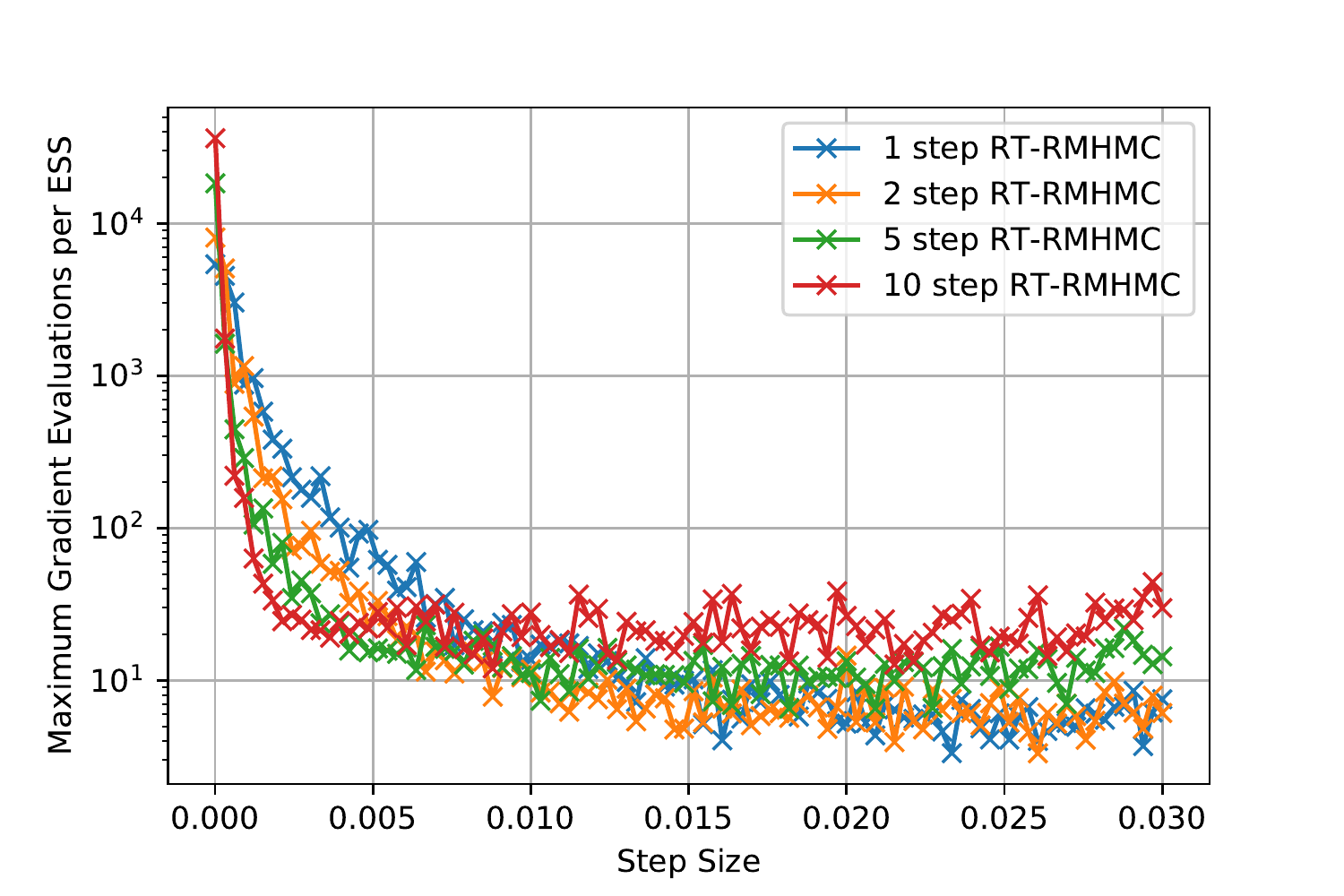}         \caption{}
	\end{subfigure}
	\hfill
	\begin{subfigure}{0.4\textwidth}
		
		\includegraphics[width=\linewidth]{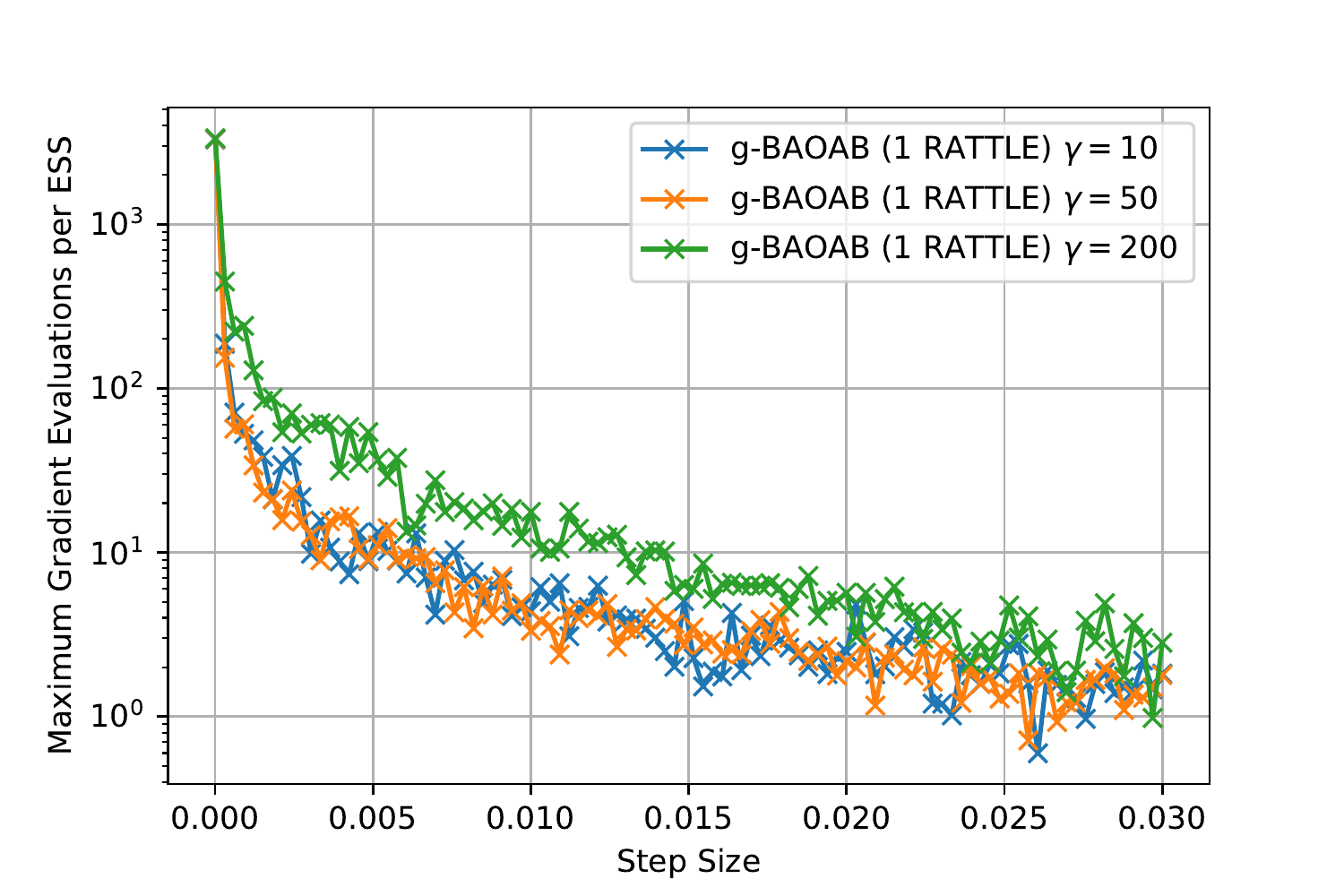}
		\caption{}
		
	\end{subfigure}
	\hfill
	\begin{subfigure}{0.4\textwidth}
		
		\includegraphics[width=\linewidth]{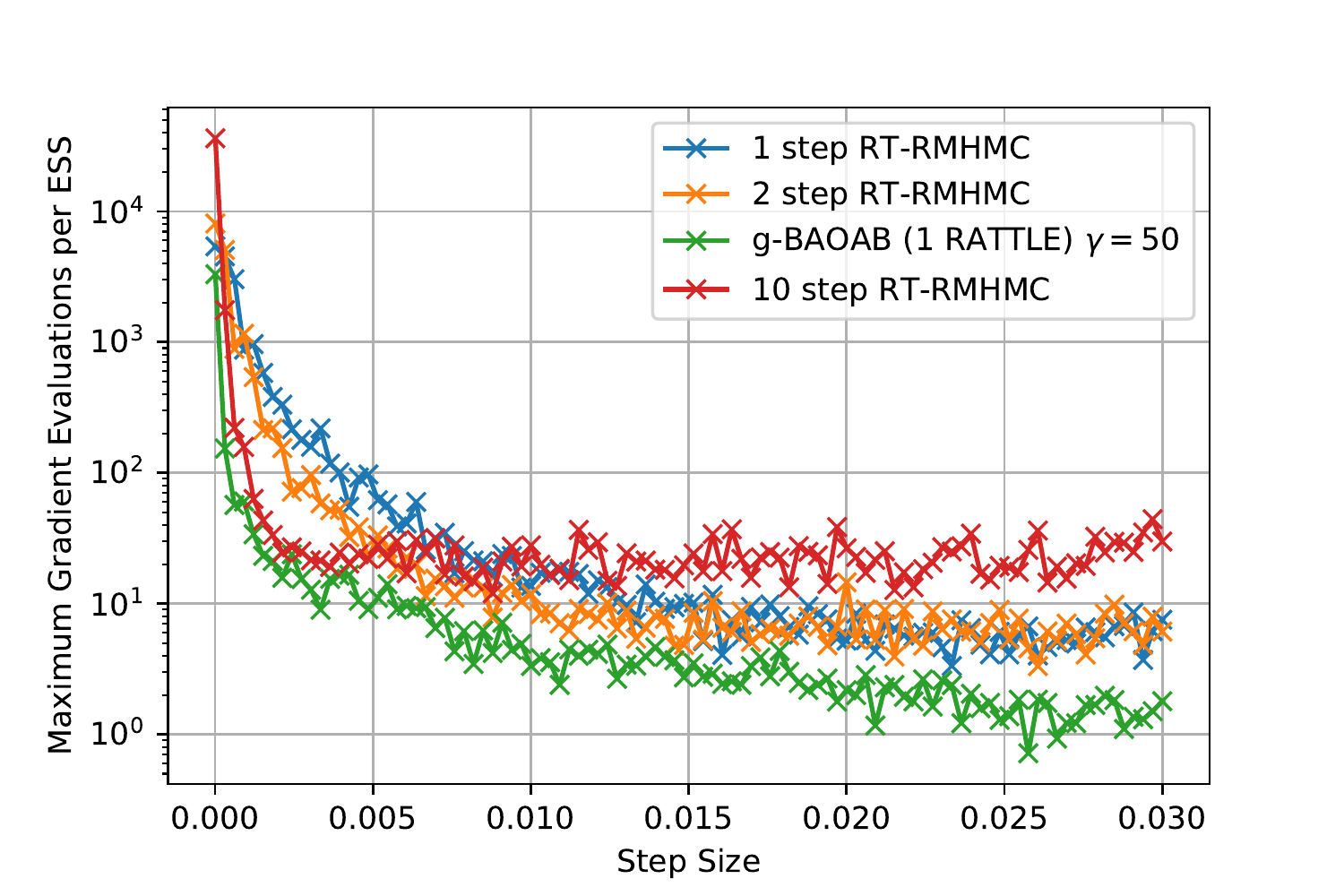}
		\caption{}
		
	\end{subfigure}
	\caption{Maximum gradient evaluation per ESS estimates over $x_1 , x_2$ and $x_3$ for the BVMF distribution using 100,000 samples with parameters $A = \text{diag}(-1000,0,1000)$ and $c = (100,0,0)$ and for varying choices of step-size. Upper left: RMHMC. Upper right: RT-RMHMC. Lower left: g-BAOAB. Lower right: g-BAOAB and RT-RMHMC.}
	\label{fig:Max-ESS-BAOAB}
	
\end{figure}

In our first example and Figure \ref{IAC-2} we can see that the regularity of the quality of samples with respect to the duration parameter is poor when a deterministic duration parameter is used and nearly uniform across a wide interval for a randomized duration with the same expected value. This is illustrated in Figure \ref{fig:Contour}, where a small change in duration parameter causes the dynamics dramatically slows convergence and due to very slow mixing. The fact that RMHMC behaves erratically for large mean duration parameters may not be very surprising to some readers as the theoretical convergence bound for HMC without randomization requires a limit on the duration $T$ (see \cite{Ma2017}). This is due to the fact that when $T$ is set too large, the coupling argument breaks down.

We next compare efficiencies using the metric gradient evaluations per effective sample size, which tells us the number of gradient evaluations needed for one independent sample in estimating our observables. We compare this metric for varying choices of step-size up to the reversibility condition is broken and the numerical integrator becomes unstable. Our observables will be $-\log \pi_{\mathcal{H}}$ and $x_{i}$ for $i = 1,...,n$. We can see as in Figure \ref{fig:ESS-BAOAB} and \ref{fig:Max-ESS-BAOAB} RMHMC (without deterministic time) exhibits the same behaviour as in Figure \ref{IAC-2} for all choices of step-size, which is not the case for RT-RMHMC. We next compare the efficiency of the method with the g-BAOAB constrained Langevin integrator, we find in Figure \ref{fig:ESS-BAOAB} and \ref{fig:Max-ESS-BAOAB} that g-BAOAB outperforms RT-RMHMC for large choices of the friction parameter $\gamma$ (for this example $\gamma = 50$). g-BAOAB has no Metropolis-Hastings adjustment and hence is a biased sampling method. The bias in the samples creates errors in computed observables. For  large choices of $\gamma$, this bias is dramatically reduced, but use of high friction may slow convergence of metastable systems.  

To explore this we next consider a bimodal distribution from \cite{By2013} in Figure \ref{fig:MC-Average}. It is shown in Figure \ref{fig:MC-Average} that g-BAOAB incurs bias for large stepsizes and convergence is slow for large choices of the friction parameter for this metastable system. The Figure also shows that this is not the case for RT-RMHMC. In Figure \ref{fig:ESS-BAOAB} we choose step-sizes up to which the integrator is reversible and stable.

\begin{figure}[H]
	\centering
	\begin{subfigure}[b]{0.32\textwidth}
		\centering
		\includegraphics[width=\textwidth]{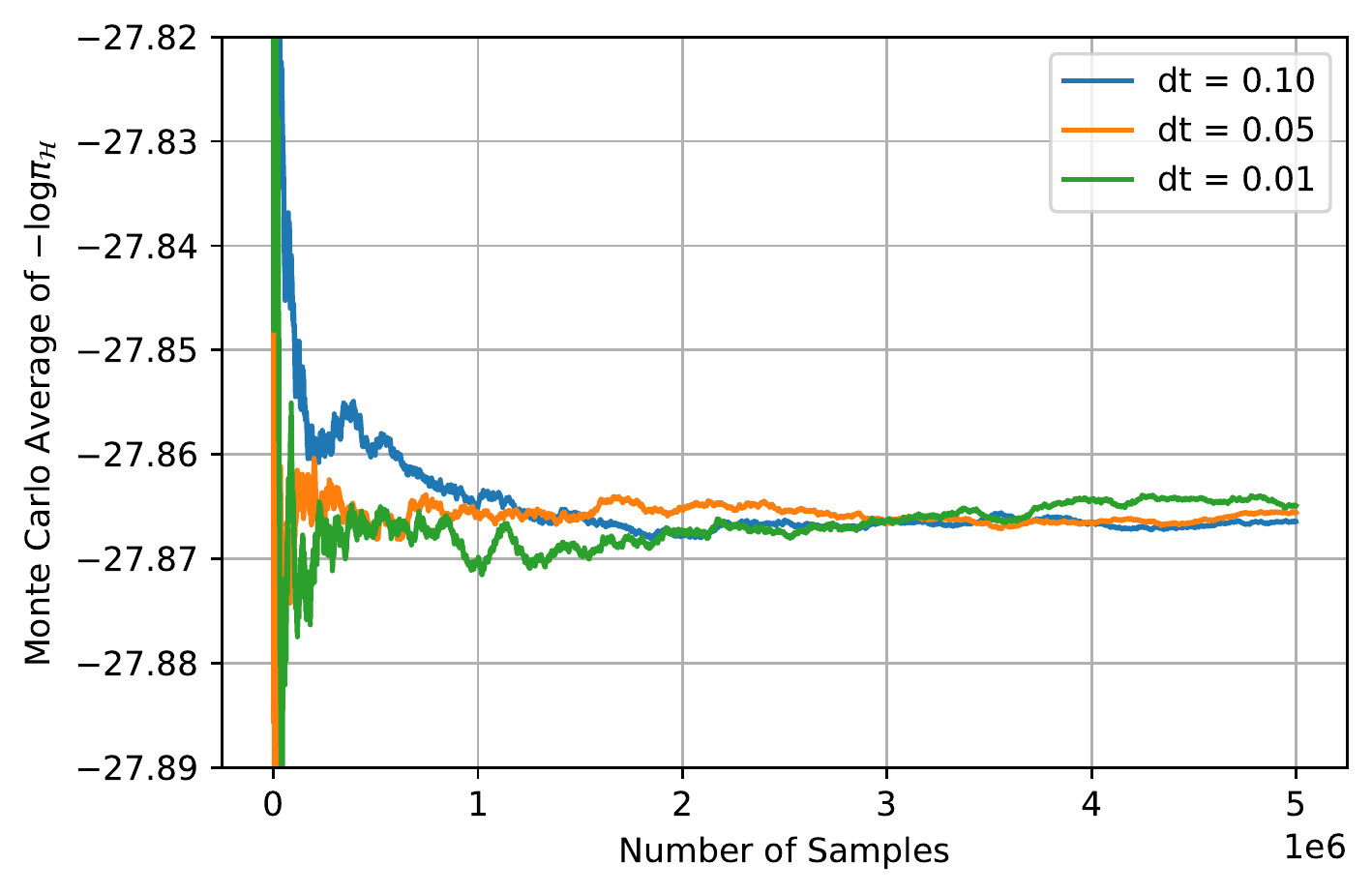}
		\caption{}
		
	\end{subfigure}
	\hfill
	\begin{subfigure}[b]{0.32\textwidth}
		\centering
		\includegraphics[width=\textwidth]{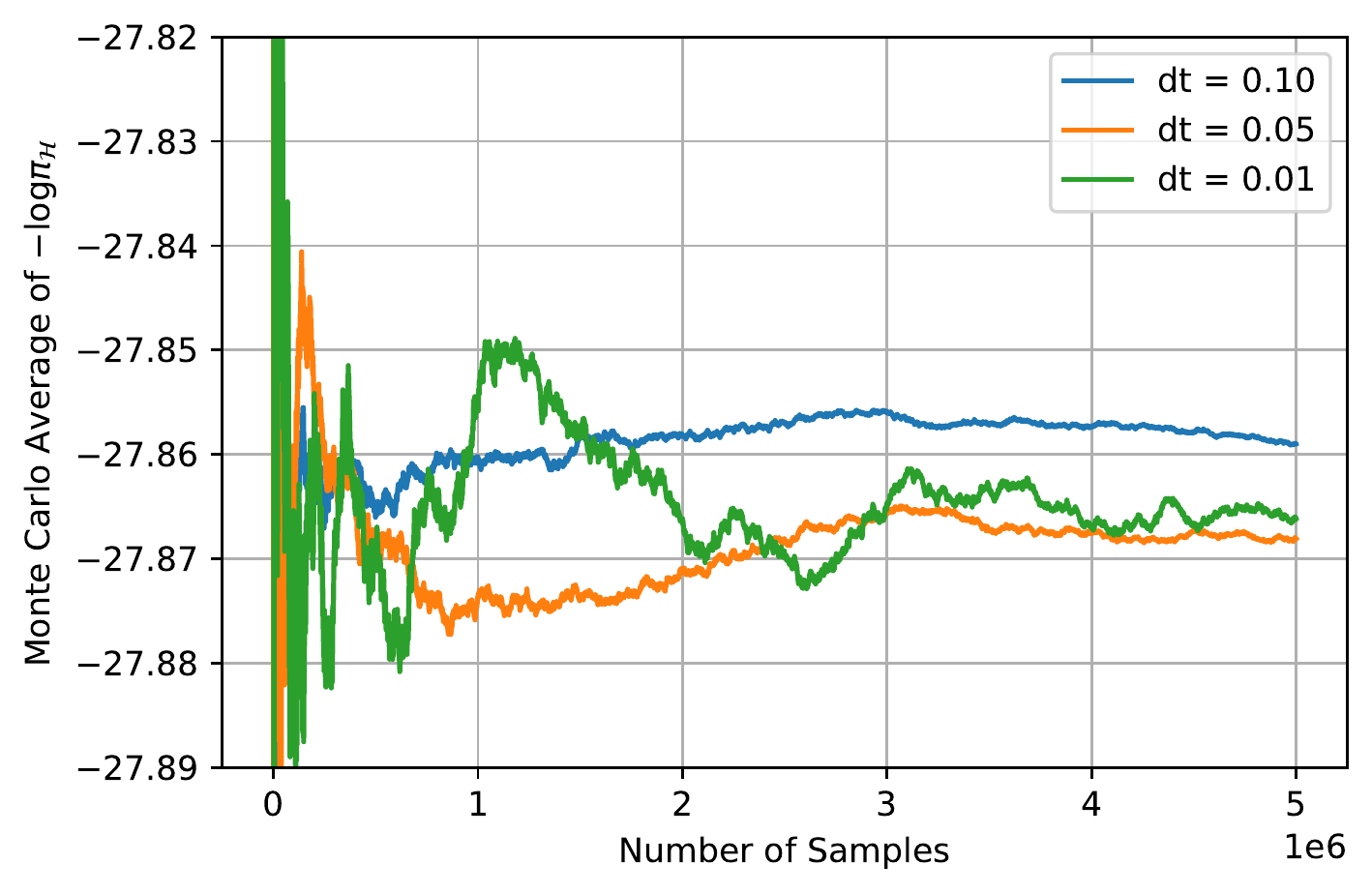}
		\caption{}
		
	\end{subfigure}
	\hfill
	\begin{subfigure}[b]{0.32\textwidth}
		\centering
		\includegraphics[width=\textwidth]{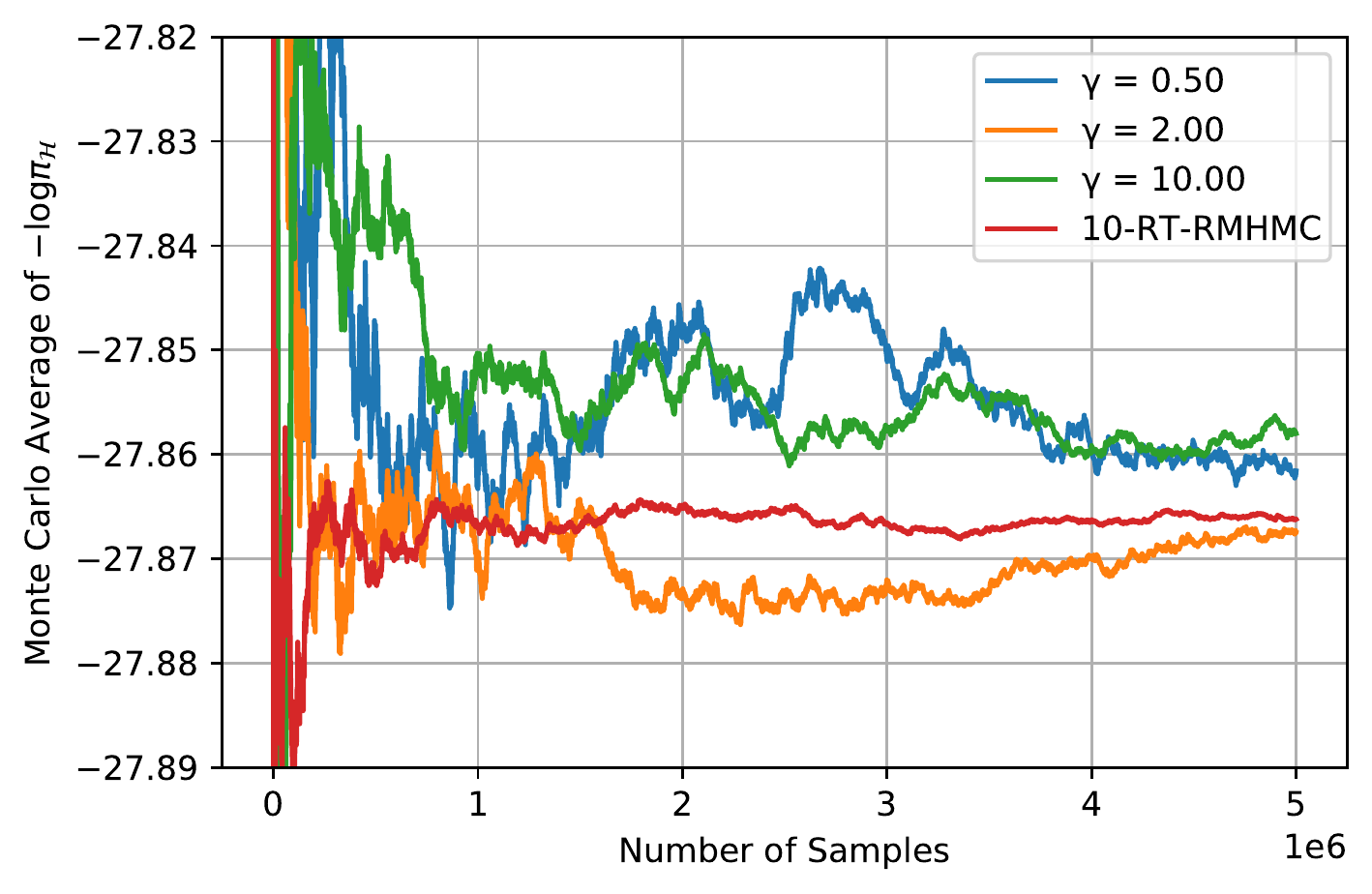}
		\caption{}
		
	\end{subfigure}
	\caption{Monte Carlo average of $-\log \pi_{\mathcal{H}}$ for the BVMF distribution with parameters $A = \textnormal{diag}(-20,-10,0,10,20)$ and $c = (40,0,0,0,0)$ with $N = 5 \times 10^6$ samples. Left: 10 step RT-RMHMC. Middle: g-BAOAB with $\gamma = 2$. Right: g-BAOAB with $dt = 0.01$}
	\label{fig:MC-Average}
\end{figure}

The efficiency of the methods with their optimal choices of parameters is comparable, but RT-RMHMC is much less sensitive with respect to the choice of parameters (stepsize and number of leapfrog steps) compared to RMHMC, so RT-RMHMC is more reliable from this point of view. This is important as it is hard to know an appropriate choice of parameters apriori and the integration length between samples might have to be arbitrarly small for RMHMC to be efficient.

\subsubsection*{Von Mises-Fisher distribution on $\mathbb{V}_{d,p}$}

\begin{definition}
	A \textit{Stiefel manifold} $\mathbb{V}_{d,p}$ is the set of $d \times p$ matrices $X$ such that $X^{T}X = I$.
\end{definition}
These arise in many statistical problems which are discussed in \cite{By2013}. Applications include dimensionality reduction such as is used in factor analysis, principal component analysis (\cite{Jo1986}) and directional statistics (\cite{Ma2000}). These are a generalisation of orthogonal groups. The von Mises-Fisher distribution on Stiefel manifolds is defined by the density 
\[p_{vMF}(X) \propto \exp{(Tr(F^{T}X))} = \exp{(\langle f_{1},x_{1} \rangle + ... + \langle f_{p},x_{p}\rangle)},\]
where $x_{i}$ and $f_{i}$ are the columns of $F$ and $X$. We simulate IAC esimates for two example distributions for varying duration parameters. In the simulations we use a stepsize of $\Delta t = 0.001$ and $100,000$ samples in each IAC estimate. The results are shown in Figure $\ref{fig:IAC-Stiefel}$. We can see similar behaviour as the easier example on the Sphere. Both examples it is clear RMHMC is much more sensitive to the mean duration and hence with respect to stepsize and number of leapfrog steps. We note that $\textnormal{Skew}(2,-45,-4)$ denotes the $3$ by $3$ skew-symmetric matrix with up triangular entries $2,-45$ and $-4$.

\begin{figure}[H]
	\centering
	\begin{subfigure}[b]{0.4\textwidth}
		\centering
		\includegraphics[width=\textwidth]{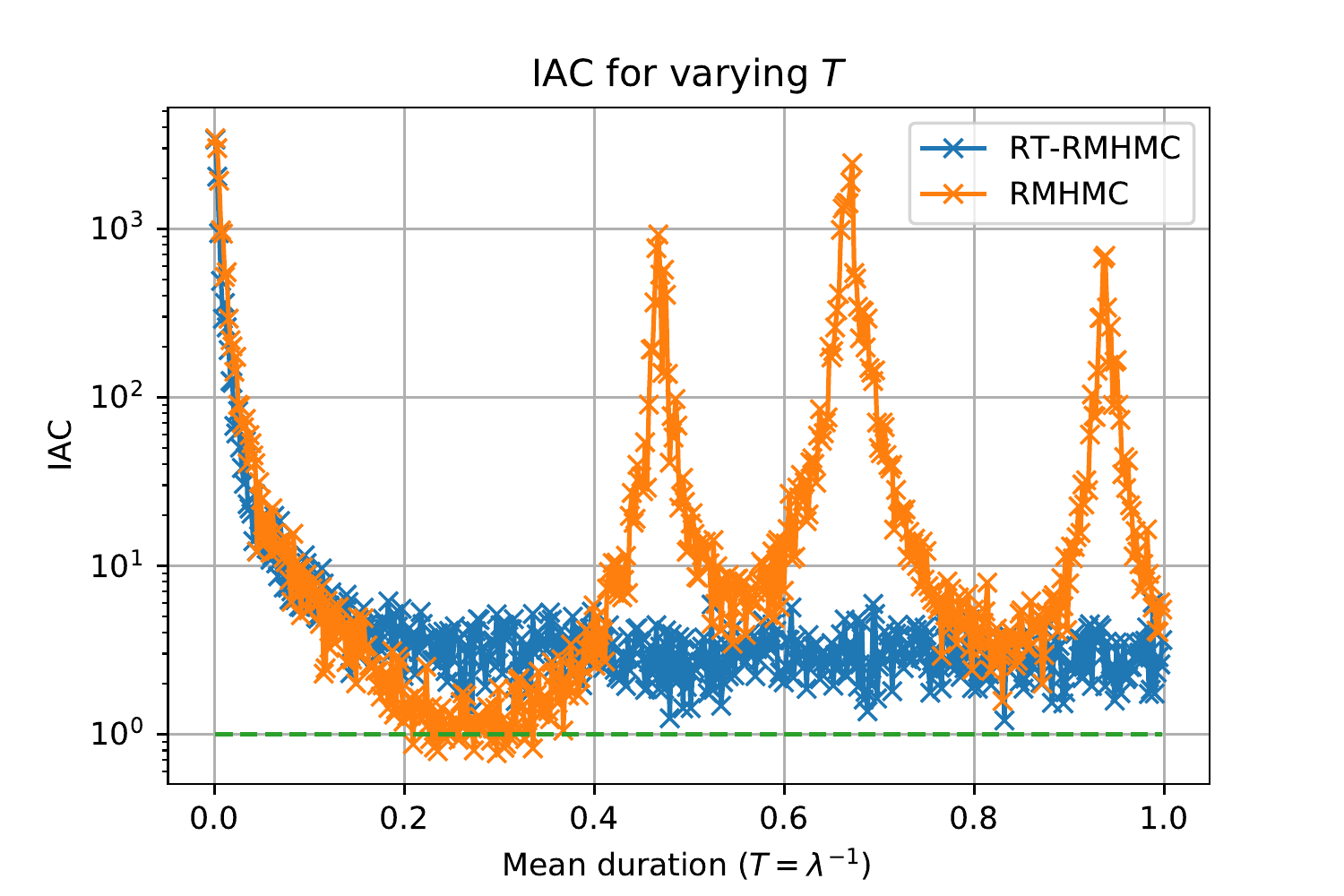}
		\caption{}
		
	\end{subfigure}
	\hspace{1cm}
	\begin{subfigure}[b]{0.4\textwidth}
		\centering
		\includegraphics[width=\textwidth]{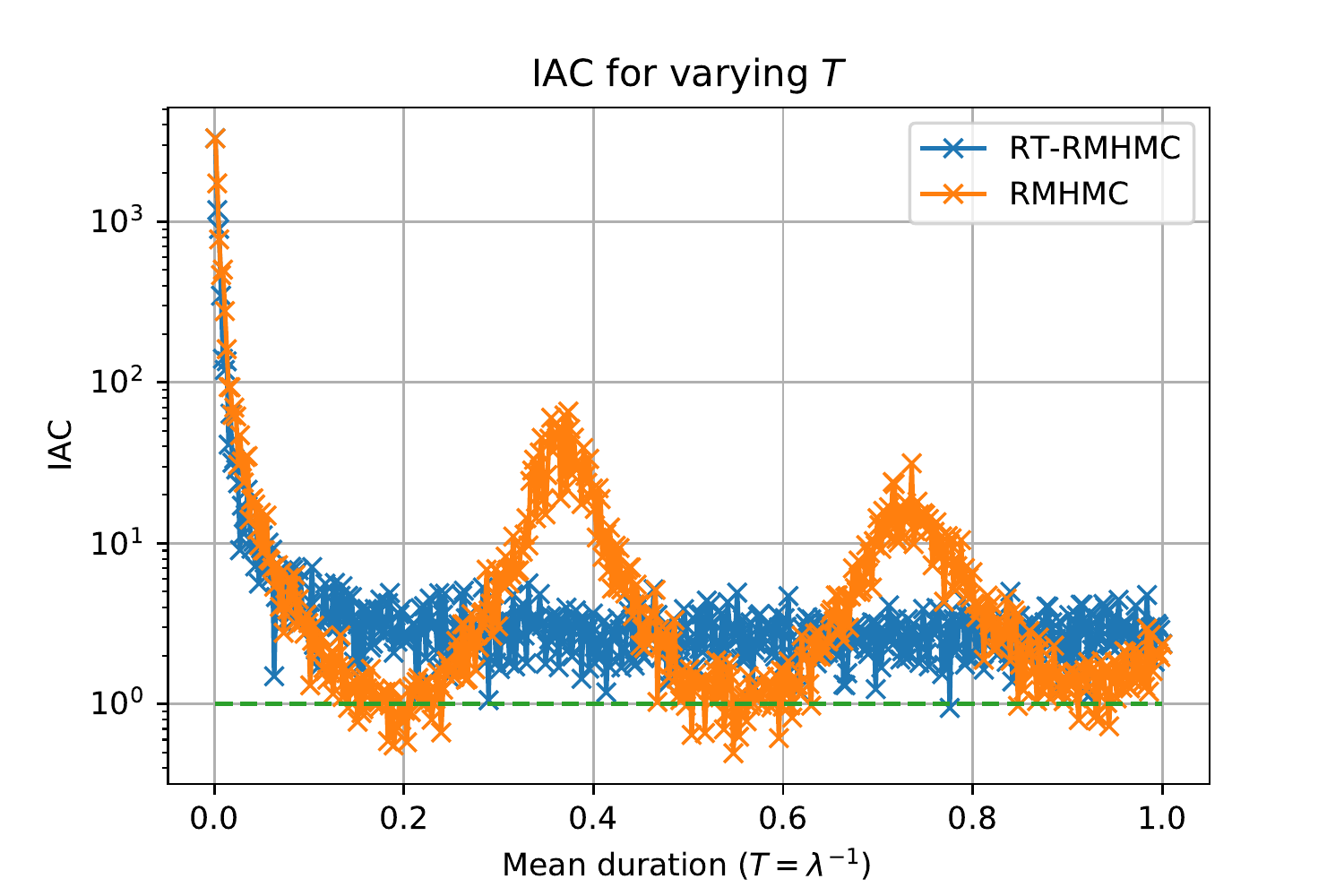}
		\caption{}
		
	\end{subfigure}
	
	\caption{IAC estimates of $-\log \pi_{\mathcal{H}}$ for different choices of $\lambda$ for the VMF distribution on $O(3)$ and $\mathbb{V}_{18,3}$ with parameters given in the captions. Left: VMF distribution on $O(3)$ with parameters $F = A := \text{Skew}(2,-45,-4)$. Right: VMF distribution on $\mathbb{V}_{18,3}$ with parameters $F = [I,-A,I,-A,I,-A]^{T}$.}
	\label{fig:IAC-Stiefel}
\end{figure}

\subsection{High Dimensional Covariance Estimation} \label{sec:high-dim-cov}

In many statistical applications for analysing high dimensional data sets it is necessary to estimate sample covariances. This can be challenging when the number of dimensions is larger than the number of data points, as the sample covariance estimator does not work well in such cases. \cite{La2020} provides a review of high-dimensional covariance estimation and applications in principal component analysis (\cite{Sh2016}), cosmological data analysis (\cite{Jo2017}) and finance (\cite{La2016}). \cite{La2020} focuses on the set up where the matrix dimension is diverging or even larger than the sample size. In this setting one needs to estimate the population covariance matrix $\Sigma$ of a set of $n$, $p-$dimensional data vectors, which we assume are drawn from an underlying distribution. 

One estimator is the sample covariance matrix, which is defined by $\Sigma_{S} = \nicefrac{1}{n}\sum^{n}_{i=1}(\mathbf{x}_{k} - \overline{\mathbf{x}})(\mathbf{x}_{k} - \overline{\mathbf{x}})^{T}$, where $\overline{\mathbf{x}} = \nicefrac{1}{n}\sum^{n}_{k = 1}\mathbf{x}_{k}$ is the sample mean. However this is a poor estimator of $\Sigma$ when $p$ is large compared to the sample size $n$ (due to rank deficiency). A way to combat this is to consider regularised covariance matrix estimators which include structural assumptions on the covariance matrix $\Sigma$. 

One such method which has been proposed assumes the structure of a low rank matrix plus a sparse matrix (see \cite{Ro2013} and \cite{La2020}). This structure is known as a spiked covariance structure and has been studied in \cite{Bouch2020},\cite{La2020} and \cite{Ca2015}. 
There have been interesting applications to finance (\cite{fan2008high}), chemometrics(\cite{kritchman2008determining}), and astronomy (\cite{Jo2017}).
The covariance matrix $\Sigma$ is assumed to be expressible in the form
\[
\Sigma = X D_{1} X^{T} + D_2,
\]
where $X$ is a Stiefel manifold of dimension $p \times m$ for $p >> m$ and $D_{i}$ for $i = 1, 2$ are diagonal matrices of dimensions $m \times m$ and $p \times p$ respectively. 
The motivation for this structure is that we assume that lower dimensional variables $\mathbf{y}_{i}$ can describe the data $\mathbf{x}_{i}$ such that $\mathbf{x}_{i} = X\mathbf{y}_{i} + \boldsymbol{\epsilon}_{i}$, where $X$ is a $p \times m$ matrix with orthogonal columns. 
We have that
\[
\Sigma = X\Sigma_y X^{T} + \Sigma_{\epsilon}, 
\]
which we interpret as a low rank matrix (rank $m$) plus a sparse matrix. We take $\Sigma_y$ and $\Sigma_{\epsilon}$ to be diagonal, which is an approximation of the spiked covariance structure \cite{Ch1982}.

Assume a uniform prior on $X$ with respect to the Hausdorff measure on $X$. Further assume a half-normal prior of the diagonal entries of $D_{i}$ for $i = 1,2$ to ensure positive definiteness. We also consider the following likelihood for the covariance estimation
\[
\mathcal{L}(\Sigma \mid \mathbf{x}_{1},...,\mathbf{x}_{p} ) = (2\pi)^{-np/2} \prod^{n}_{i=1} \det{(\Sigma)}^{-1/2} \exp{\left(-\frac{1}{2}(\mathbf{x}_{i} - \overline{\mathbf{x}})^{T}\Sigma^{-1}(\mathbf{x}_{i} - \overline{\mathbf{x}})\right)}. 
\]
We introduce the posterior distribution $p(\Sigma \mid \mathbf{x}_{1},...,\mathbf{x}_{p}) \propto \mathcal{L}( \Sigma \mid \mathbf{x}_{1},...,\mathbf{x}_{p}) p(\Sigma)$, where $p(\Sigma) = p(X)p(D_{1})p(D_{2})$ for $X \sim \mathcal{U}(\mathbb{V}_{p,m})$, $D_{1_{jj}} \sim \mathcal{N}_{+}(0,\sigma^{2}_{1})$ for $j = 1,...,m$ and $D_{2_{jj}} \sim \mathcal{N}_{+}(0,\sigma^{2}_{2})$ for $j = 1,...,p$ and $\mathcal{N}_{+}$ denotes the half-normal distribution. Define the potential $U: \mathbb{V}_{p,m} \times \mathbb{R}^{m} \times \mathbb{R}^{p} \mapsto \mathbb{R}$ by $U(X,\mathbf{d}_{1},\mathbf{d}_{2}) = -\log{\mathcal{L}( \Sigma(X,\mathbf{d}_{1},\mathbf{d}_{2}) \mid \mathbf{x}_{1},...,\mathbf{x}_{p})} -\log{p(X)} - \log{p(\mathbf{d}_{1})}-\log{p(\mathbf{d}_{2})}$ with forces given by
\[
\frac{\partial U}{\partial X_{ij}} = \frac{\partial U}{\partial \Sigma_{kl}}\frac{\partial\Sigma_{kl}}{\partial X_{ij}} = \left( \frac{1}{2}n(\Sigma^{-1})^{T}_{kl} + \frac{1}{2}\sum^{n}_{r = 1}(\mathbf{x}_{r} - \overline{\mathbf{x}})^{T}B^{kl}(\mathbf{x}_{r} - \overline{\mathbf{x}})\right) \frac{\partial\Sigma_{kl}}{\partial X_{ij}}, 
\]
where $[B^{kl}]_{ij} = -(\Sigma^{-1})_{ik}(\Sigma^{-1})_{lj}$ and 
\[
\frac{\partial\Sigma_{kl}}{\partial X_{ij}} = \begin{cases} X_{kj}D^{1}_{jj} & k \neq i, l = i,\\
D^{1}_{jj}X_{lj} & k = i, l \neq i, \\
2X_{ij}D^{1}_{jj} & k = i, l = i,\\
0 & \textnormal{otherwise.}
\end{cases}
\]
We also have that
\[
\frac{\partial U}{\partial d_{i_{j}}} = \frac{\partial U}{\partial \Sigma_{kl}}\frac{\partial\Sigma_{kl}}{\partial d_{i_{j}}} + \frac{ d_{i_{j}}}{\sigma^{2}_{1}},
\]
for $i = 1,2$ and where 
\[
\frac{\partial\Sigma_{kl}}{\partial d_{1_{j_1}}} = X_{kj_1}X_{lj_1} \quad \text{and} \quad \frac{\partial\Sigma_{kl}}{\partial d_{2_{j_2}}} = \begin{cases} 1 & \text{if } k = l = j_2\\
0 & \text{otherwise}
\end{cases}
\]
for $j_1 = 1,...,m$ and $j_2 = 1,...,p$. We will use the likelihood and its gradient for implementing our RT-RMHMC algorithm for such models.
\subsubsection*{Covariance estimation for cosmological data}

We consider an application of high dimensional covariance estimation in cosmological data analysis introduced in \cite{Jo2017} and discussed in \cite{La2020}. The data is taken from \cite{Jo2017} and consists of covariances of two-point correlation functions of cosmic weak lensing. It is simulated using coupled log-normal random fields from angular power spectra. For further information we refer the reader to \cite{Jo2017}. We will test this method using $2n/3$ data vectors, where $n$ is the dimension of the data vectors. Therefore we are in the setting where the dimension of the covariance matrix is larger than the number of samples. For our low rank plus sparse structure we choose $m = p/6 < < p$ and to ensure fast convergence we will normalize the data entry-wise and initialise our Markov chain via a eigenvalue decomposition of the sample covariance. We initialise our Markov chain as $\Sigma_{0} = X^{T}D_{1}X + D_{2} \approx \Sigma_{S}$, using the sample covariance matrix $\Sigma_{S}$ and where $D_2$ is the diagonal of $\Sigma_{S}$ and $X^{T}D_{1}X$ corresponds to the eigenvalue decomposition of $\Sigma_{S} - D_{2}$, but with the $p$ largest eigenvectors. After the covariance of the normalized data is estimated it can easily be rescaled to match the real data via entry-wise multiplication with the outer product of the entry-wise standard deviations. We compare our method to a maximum a posterior (MAP) estimate of the covariance matrix which uses a simple constrained gradient descent algorithm with Lagrange multipliers to ensure that the ``low-rank plus sparse'' structure is maintained. We compare the Bayesian and MAP approaches using a relative Frobenius norm and a covariance metric introduced by \cite{Fo2003} which is defined by
\[
d(A,B) = \sqrt{\sum^{n}_{i=1} \ln^{2}{\lambda_{i}(A,B)}} ,
\]
where $A$ and $B$ are covariance matrices and $\lambda_{i}(A,B)$ are the generalized eigenvalues from ${\rm det}(\lambda A - B)= 0$. As pointed out in \cite{Fo2003}, this covariance metric is affine invariant and invariant to inversion.

\begin{figure}
	\centering
	\begin{subfigure}{0.32\textwidth}
		
		\includegraphics[width=\linewidth]{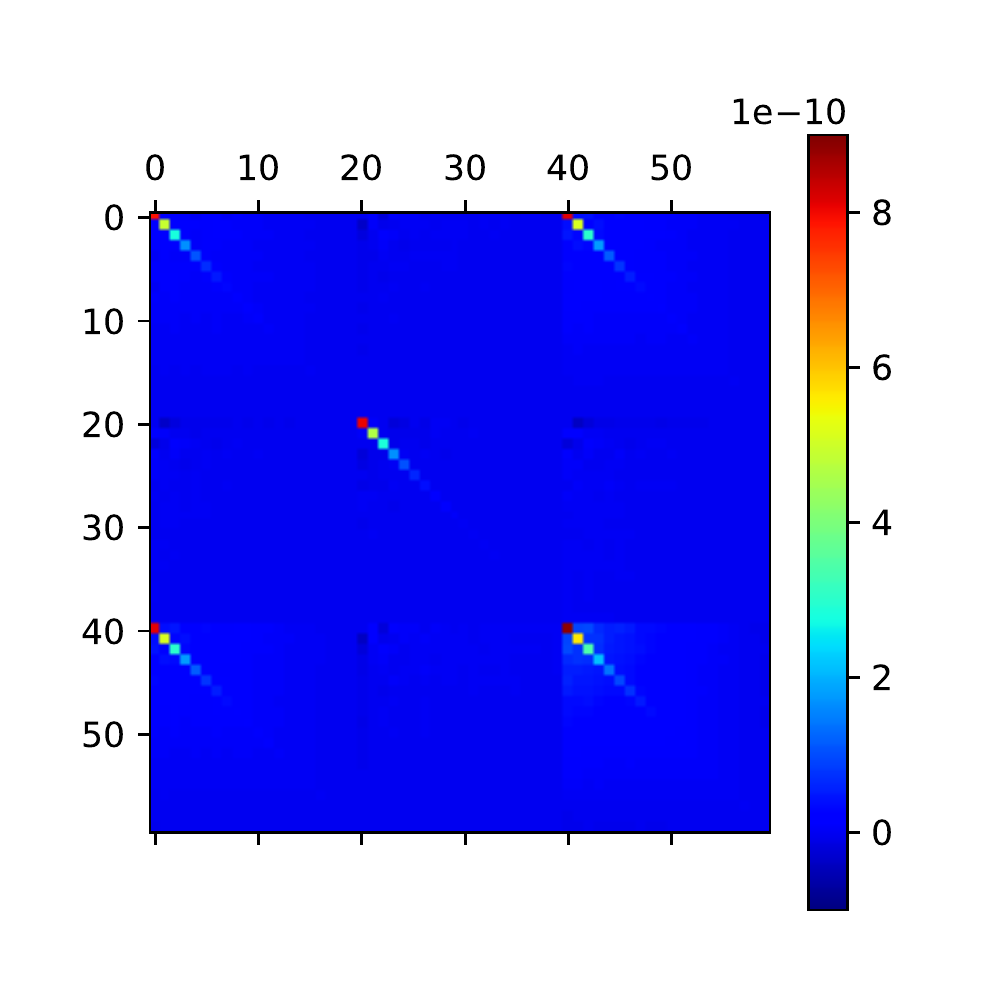}
		\caption{}
	\end{subfigure}
	\begin{subfigure}{0.32\textwidth}
		
		\includegraphics[width=\linewidth]{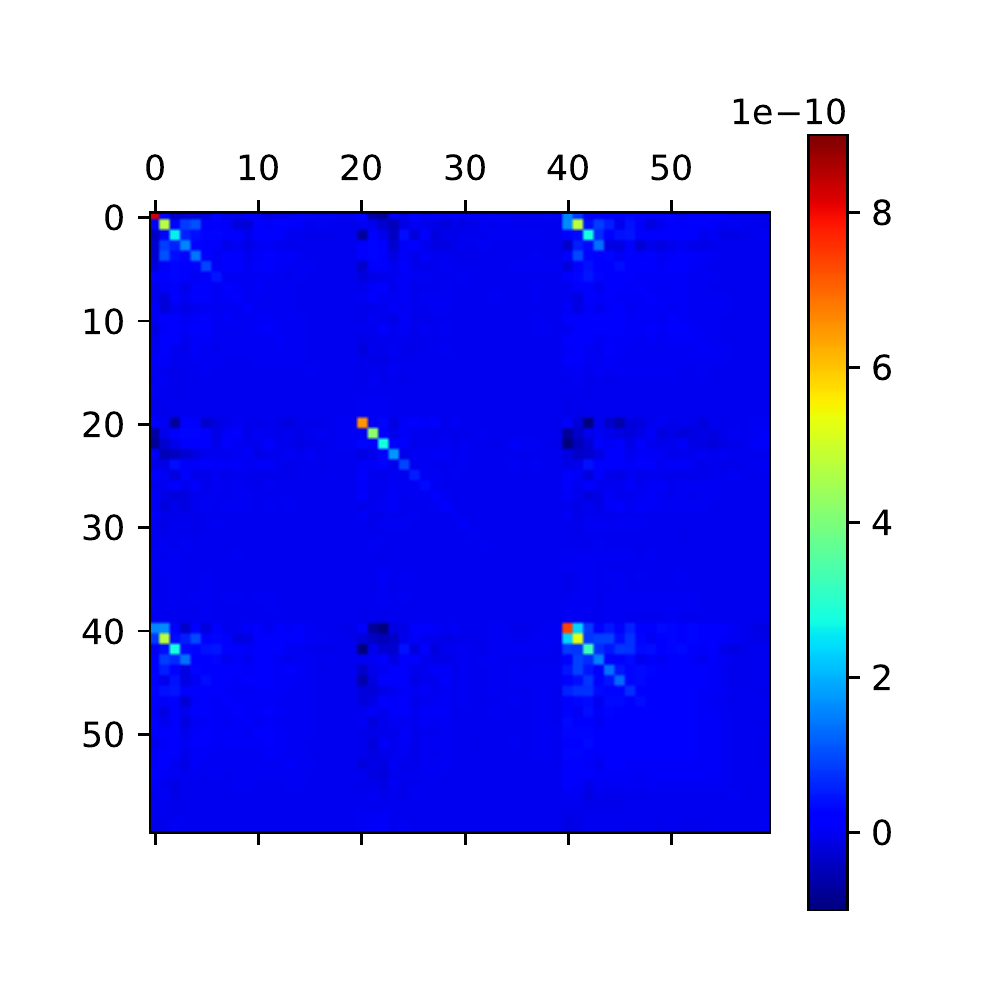}
		\caption{}
		
	\end{subfigure}
	\begin{subfigure}{0.32\textwidth}
		
		\includegraphics[width=\linewidth]{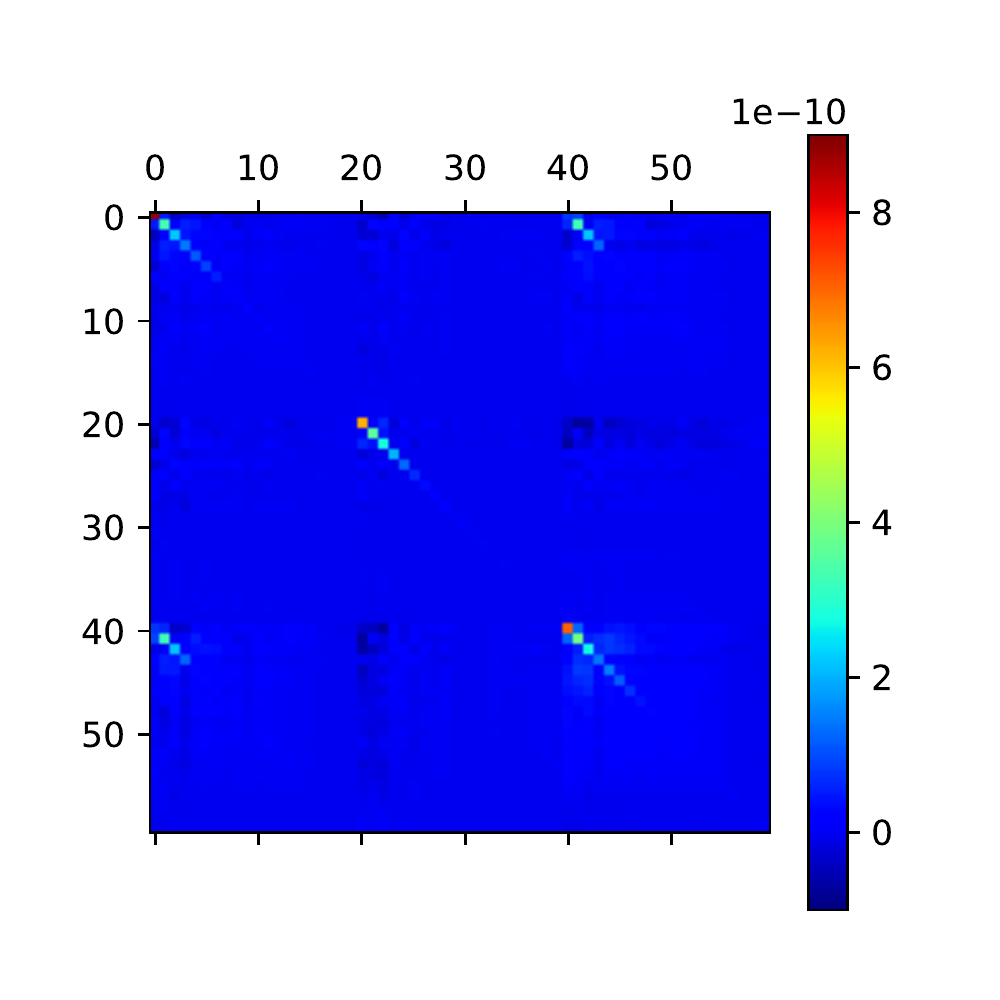}
		\caption{}
		
	\end{subfigure}
	
	\begin{subfigure}{0.32\textwidth}
		
		\includegraphics[width=\linewidth]{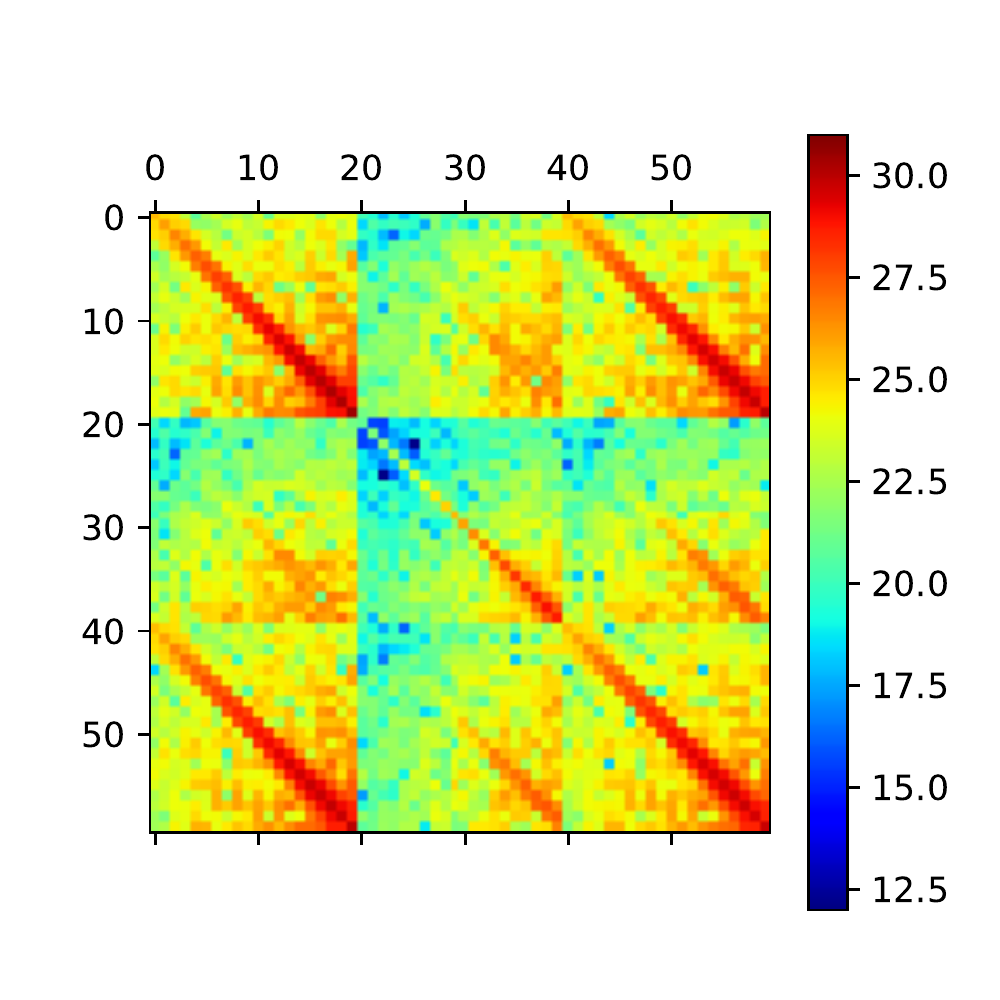}         
		\caption{}
	\end{subfigure}
	\begin{subfigure}{0.32\textwidth}
		
		\includegraphics[width=\linewidth]{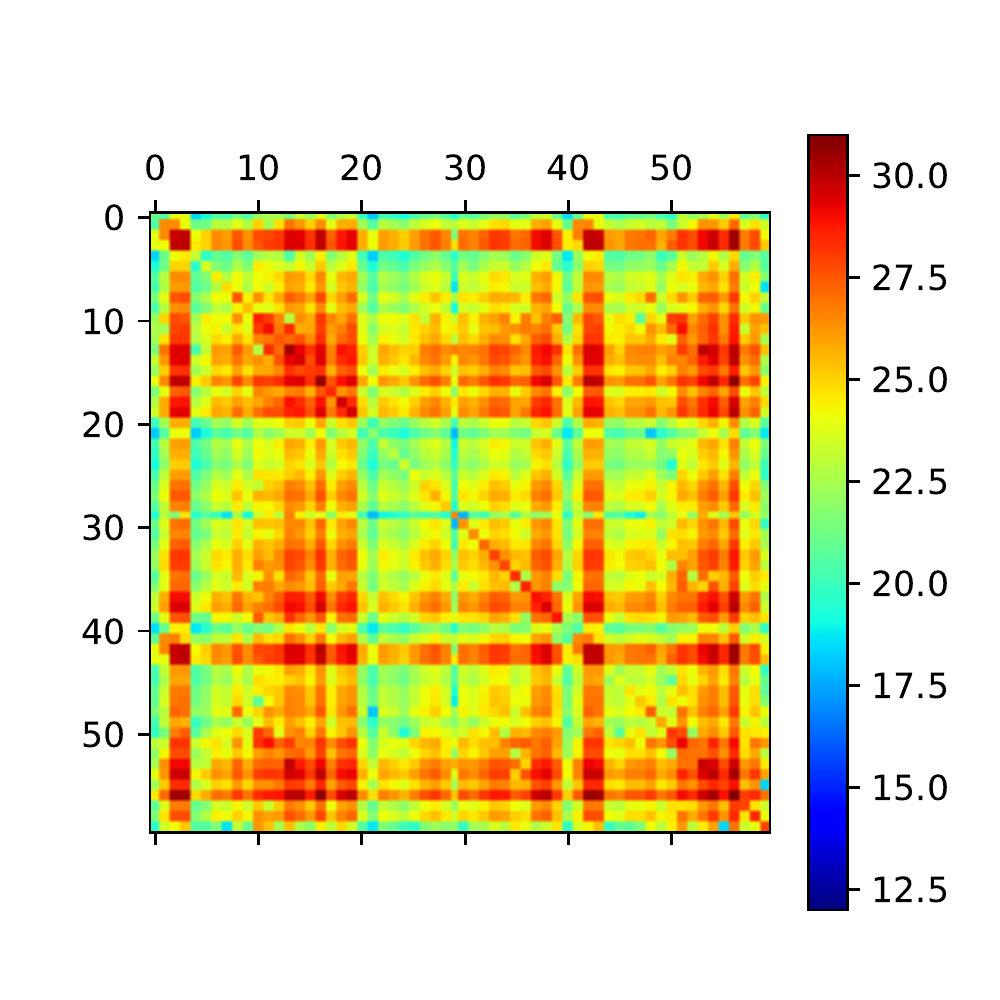}         
		\caption{}
	\end{subfigure}
	\begin{subfigure}{0.32\textwidth}
		
		\includegraphics[width=\linewidth]{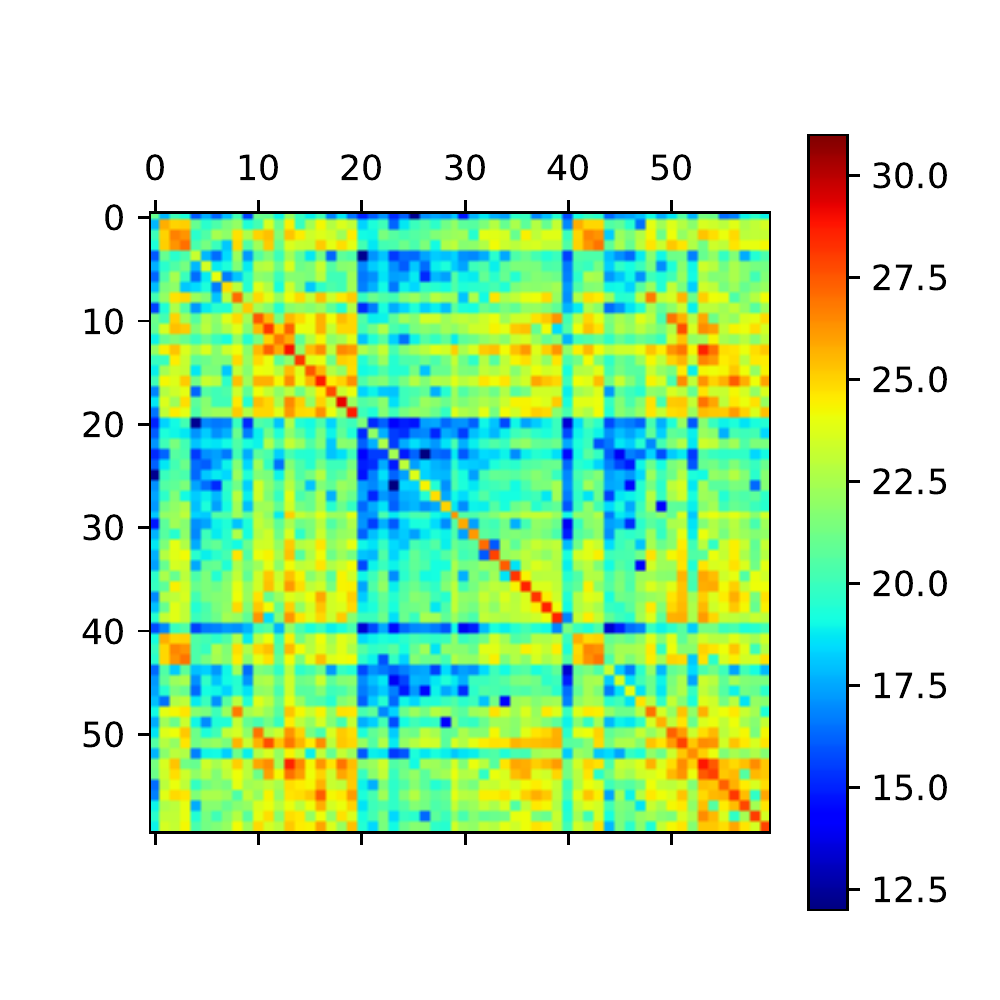}
		\caption{}
		
	\end{subfigure}
	\caption{Covariance estimates of astronomical data of $60$-dimensional data vectors. $\Sigma$ is the true covariance using all $2000$ data points. $\Sigma_{MAP}$ is a maximum a posteriori covariance estimate and $\hat{\Sigma}$ is a posterior expectation estimate of $\Sigma$ using $40$ data points. $\hat{\Sigma}^{-1}$ is a posterior expectation estimate of $\Sigma^{-1}$ using $40$ data points . The posterior expectation estimate uses priors of $\sigma_1 = \sigma_2 = 2$ after normalisation and $500,000$ samples from RT-RMHMC. Top left: $\Sigma$. Top middle: $\Sigma_{MAP}$. Top right: $\hat{\Sigma}$. Bottom left: $\ln{\lvert\Sigma^{-1}\vert}$. Bottom middle: $\ln{\lvert\Sigma_{MAP}^{-1}\vert}$. Bottom right: $\ln{\lvert\hat{\Sigma}^{-1}\vert}$.}
	\label{fig:Covariance}
\end{figure}

\begin{table*}
    \caption{A comparison using three metrics between the MAP estimate and the posterior expectation estimate using 500,000 samples from RT-RMHMC. The relative frobenius norm of the estimate and its inverse and the covariance metric introduced by \cite{Fo2003}.}
	\centering
	\label{fig:Covariance-table}
	\begin{tabular}{@{}lrr@{}}
		\hline
		 \multicolumn{1}{c}{Metric}
		& \multicolumn{1}{c}{MAP} & \multicolumn{1}{c}{Posterior Expectation}\\
		\hline
		$\lVert\Sigma - \Sigma_{EST}\Vert_{F}/\lVert\Sigma\Vert_{F}$  &  $0.4862$   & $0.5155$\\
		\hline
		$\lVert\Sigma^{-1} - \Sigma^{-1}_{EST}\Vert_{F}/\lVert\Sigma^{-1}\Vert_{F}$   & $1.8503$   & $0.8946$\\
		\hline
		$d(\Sigma,\Sigma_{EST})$  &  $15.2515$  & $14.7547$  \\
		\hline
	\end{tabular}
\end{table*}

Note that in Figure \ref{fig:Covariance} we have not included the sample estimate with $40$ samples because the sample covariance matrix is rank deficient and hence it is not possible to invert the matrix. We notice from Figure \ref{fig:Covariance} and table \ref{fig:Covariance-table} that the MAP and Posterior Expectation perform well when estimating the covariance matrix for only using 40 data points, but lose accuracy under inversion. The posterior expectation using a sampling method seems to retain more structure when inverted and provides a more accurate estimate according to the metric of \cite{Fo2003}. 

It is clear from Table \ref{fig:Covariance-table} that using the posterior means do not sacrifice accuracy compared to using the MAP estimators. An additional benefit of the Bayesian approach is that we can compute posterior standard deviations for each component of the covariance estimator, which gives error estimates. This is illustrated in Figure \ref{fig:Standard-Deviations}. By comparing these standard deviations with the covariance estimates, we can get a sense of the relative error we are making. This information can be useful when deciding on the number of data samples we need to get a satisfactory level of accuracy in estimating the covariance matrix.
Since in practice we do not have access to the true covariance matrix, there is no straightforward way to compute error estimates based on the MAP estimator, and it is challenging to see whether we have reached sufficient accuracy.
\begin{figure}
	\centering
	\begin{subfigure}[b]{0.32\textwidth}
		\centering
		\includegraphics[width=\textwidth]{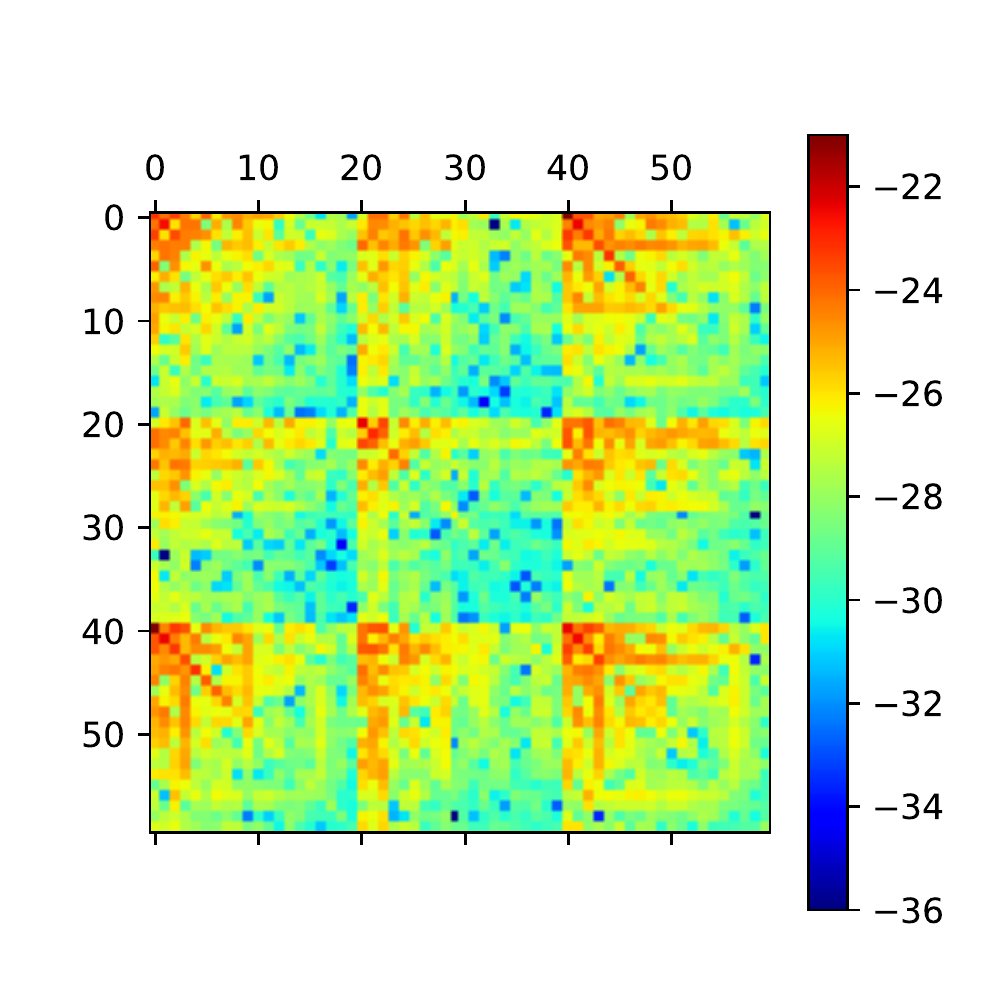}
		\caption{}
		
	\end{subfigure}
	\hspace{1cm}
	\begin{subfigure}[b]{0.32\textwidth}
		\centering
		\includegraphics[width=\textwidth]{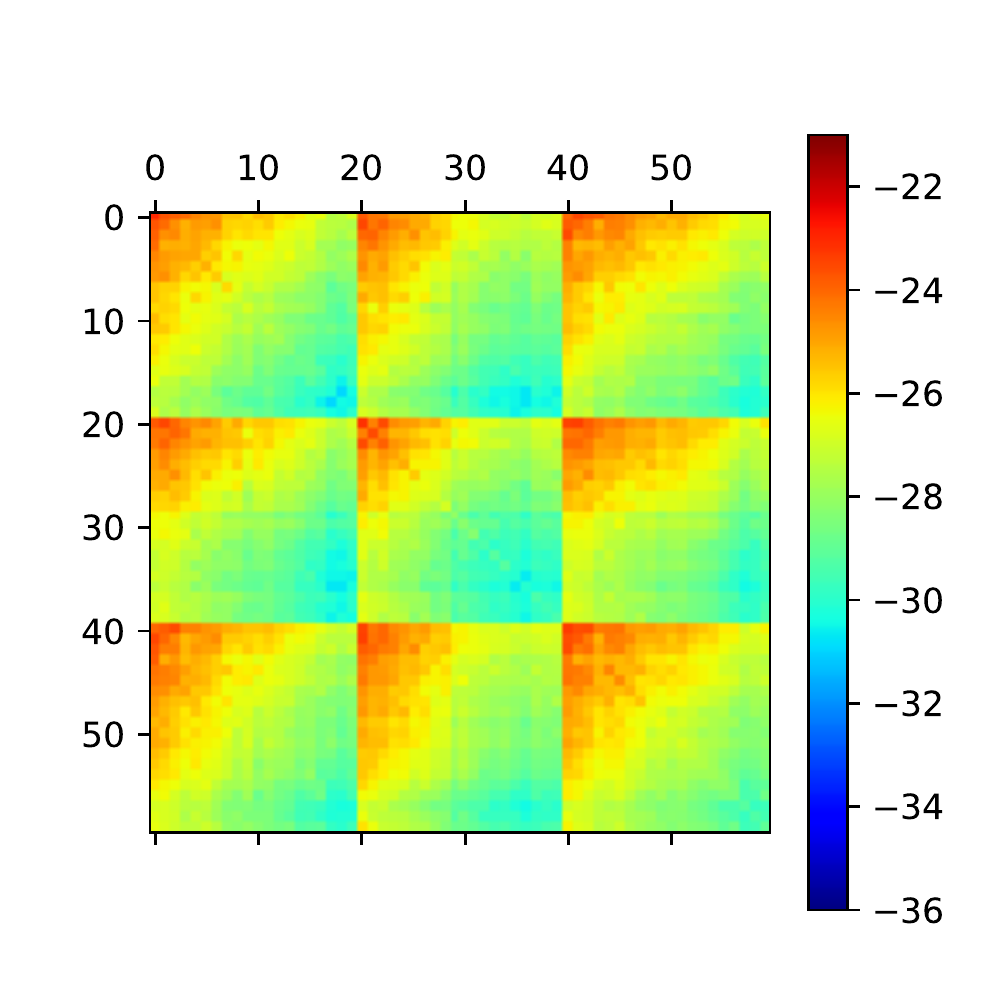}
		\caption{}
		
	\end{subfigure}
	
	\caption{A comparison between the log of the absolute error and the log of the posterior standard deviations in each component. Left: $\ln\lvert\Sigma - \hat{\Sigma}\vert$. Right: $\ln{\Sigma_{SD}}$.}
	\label{fig:Standard-Deviations}
\end{figure}


\section{Conclusion and Future Work}

In this work we have introduced a Randomized Time Riemannian Manifold Hamiltonian Monte Carlo (RT-RMHMC), which is a robust alternative to Riemannian Manifold Hamiltonian Monte Carlo methods introduced by \cite{Gi2011} and \cite{Br2012}. We establish invariance of the desired measure under a compactness assumption in the continuous (small stepsize limit) setting. We provide an Metropolis adjusted version of RT-RMHMC in the discrete setting and prove invariance and ergodicity of the adjusted discretized algorithm. We show that RT-RMHMC is a more robust method with respect to parameter choice on a number of numerical examples arising in applications and  provide an example to demonstrate that our Riemannian manifold sampling method can be used for high-dimensional covariance estimation. We expect the stability with respect to choice of parameters is especially needed in poorly conditioned problems, where RMHMC would require very short time steps for stability but this may lead to some random walk behaviour and highly inefficient mixing in some principal directions.

In terms of future developments for RT-RMHMC, the next step would be to establish invariance of the measure in the non-compact setting and further to this establish (geometric) ergodicity  of RT-RMHMC, which is already established in \cite{Bo2017} for the Euclidean setting. Then one could find optimal choices of integration parameters and step-size. Another possibility would be to establish mixing time guarantees for RT-RMHMC by a coupling argument like \cite{Bo2020b} and \cite{Ma2018b}. In \cite{Ma2018b} they establish rapid mixing guarantees for a geodesic walk algorithm on manifolds with positive curvature, which is RMHMC for the uniform distribution. One may be able to use a similar coupling argument to guarantee mixing times for RT-RMHMC for manifolds with positive curvature.


\backmatter





\bmhead{Acknowledgments}

The authors acknowledge the support of the Engineering and Physical Sciences Research Council Grant EP/S023291/1 (MAC-MIGS Centre for Doctoral Training).










\begin{appendices}

\section{Generator of RT-RMHMC} \label{PDMPs-and-their-invariant-measures}
To prove that the generator of this stochastic process takes the form of equation \eqref{generator} and that the measure (equation \eqref{measure}) is invariant under RT-RMHMC we use the framework of \cite{Du2021} viewing RT-RMHMC as a piecewise deterministic Markov process (PDMP) defined on $T \mathcal{M}$. A $\mathcal{Z}-$valued continuous-time PDMP $(\varphi,\lambda,Q)$ consists of the following components:
\begin{itemize}
    \item a differential flow $\varphi$ on $\mathcal{Z}$ which satisfies the semi group property and is measurable. Moreover, is continuously differentiable with respect to time and a $C^{1}$-diffeomorphism of $\mathcal{Z}$.  
    \item an event rate $\lambda: \mathcal{Z} \to \mathbb{R}^{+}$, which is measurable and locally bounded.
    \item a inhomogeneous Markov transition kernel $Q: \mathbb{R}_{+} \times \mathcal{Z} \times \mathcal{B}(\mathcal{Z}) \to [0,1]$, such that for all $A \in \mathcal{B}(\mathcal{Z})$, $(t,z) \mapsto Q(t,z,A)$ is measurable and for all $(t,z) \in \mathbb{R}_{+} \times T\mathcal{M}$, $Q(t,z,\cdot) \in \mathcal{P}(\mathcal{Z})$,
\end{itemize}
where $\mathcal{B}(\mathcal{Z})$ denotes the $\sigma$-algebra on the space $\mathcal{Z}$ and $\mathcal{P}(\mathcal{Z})$ denotes the space of probability measures on the space $\mathcal{Z}$. For RT-RMHMC we consider $\mathcal{Z} = T\mathcal{M}$.
\begin{definition}
For a PDMP $Z = (Z_{t})_{t\geq 0}$, we call $\tau_{\infty}(Z) = \inf \{t \geq 0 \mid Z_{t} = \infty \}$ the explosion time of the process $(Z_t)_{t \geq 0}$. A process $(Z_{t})_{t \geq 0}$ is said to be non-explosive if $\tau_{\infty}(Z) = + \infty$ almost surely. PDMP characteristics are said to be non-explosive if for all initial distribution the associated PDMP is non-explosive.
\end{definition}

Due to the event rate $\lambda$ of RT-RMHMC being constant and bounded we have that RT-RMHMC is non-explosive. As RT-RMHMC is a non-explosive PDMP we can use the theory of \cite{Du2021}[Section 7 and 8] to estabilish the generator and invariance of the desired measure. 

Under the assumption that the expected number of events in any unit time interval $[0,t]$ is finite, it is shown in \cite{Da1993}[Theorem 26.14] that for a non-explosive PDMP with generator $\mathcal{A}$ with domain $D(\mathcal{A})$ that all $f \in D(\mathcal{A})$ and $x \in T\mathcal{M}$,
\[\mathcal{A}f(x) = D_{\varphi}f(x) + \lambda(x)(Qf(x) - f(x)). \]

If we let $N_{t}$ denote the number of events in the interval $[0,t]$ then we have for RT-RMHMC $\mathbb{E}_{x}(N_{t}) = \lambda t < \infty$, $\mathbb{E}_{x}$ denoting the expected value given the stochastic process starts with initial condition $x$. Therefore all the assumptions are satsified of \cite{Da1993}[Theorem 26.14] and \cite{Du2021}[Section 7] and we have that the generator of RT-RMHMC is given by
\[\mathcal{L}f(x) = X_{H}(f) + \lambda(Qf(x) - f(x)), \]
where $X_{H}$ is the Hamiltonian vector field and $Q$ is the transition kernel for the Gaussian distribution induced by the metric $G(x)$ on the tangent space of $x \in \mathcal{M}$.

\section{Invariant measure}\label{invariant-measure}

To prove that $\mu$ is an invariant measure of RT-RMHMC it is sufficent to show that $\int \mathcal{L}f(x) d\mu = 0$ for all $f \in D(\mathcal{L})$. As it is difficult to consider $D(\mathcal{L})$, one approach is to show that $C^{1}_{c}(T\mathcal{M})$ is a core of the generator and that $\int \mathcal{L}f(x)d\mu = 0$ for all $f \in C^{1}_{c}(T\mathcal{M})$, where $C^{k}_{c}(T\mathcal{M})$ denotes the space of $k$ times differentiable functions $f: T\mathcal{M} \to \mathbb{R}$ with compact support. 

\begin{theorem}[Infinitesimal Invariance of RT-RMHMC] \label{Infinitesimally-Invariant}
Let $\mathcal{M}$ be a smooth Riemannian manifold with metric $g$ and let $(P_{t})_{t \geq 0}$ be the semigroup of RT-RMHMC defined on $\mathcal{M}$ with potential $U \in C^{2}(\mathcal{M})$ and Hamiltonian $H = U + K \in C^{2}(T\mathcal{M})$. Let $\mu$ be the measure on $(T\mathcal{M}, \mathcal{B}(T\mathcal{M}))$ defined by
\[ \mu(dz) \propto e^{-H(x,v)}d\lambda_{T \mathcal{M}}(z),\]
where $d\lambda_{T\mathcal{M}}$ is the Liouville measure of $T\mathcal{M}$.
Then for all $f \in C^{1}_{c}(T\mathcal{M})$
\[\int_{T\mathcal{M}}\mathcal{L}f(x,v) \mu(dz) = 0, \]
where $\mathcal{L}$ is the generator of RT-RMHMC.
\end{theorem}
\begin{proof}
We have that
\begin{align*}
    \int_{T \mathcal{M}} \mathcal{L}f(x,v) \mu(dz) &= \int_{T \mathcal{M}}X_{H}(f)d \mu + \lambda\int_{T \mathcal{M}}(Q-I)f d\mu.
\end{align*}
We will now consider these two integrals separately. Considering the first integral, due to the fact that $\mu$ is a Liouville measure, $\mu$ is invariant under the Hamiltonian flow by Liouville's theorem and hence the first integral is identically zero. Now considering the second integral we have
\begin{align*}
    \int [Q&f(x,v) - f(x,v)] \mu(dz) =  \int \! \int C(x)  e^{-\frac{1}{2}\xi^{T}G(x)^{-1}\xi}(f(x,\xi)-f(x,v))d\xi e^{-H(x,v)}d\lambda_{TM}(z)\\
    &= \int \! \int C(x) e^{-\frac{1}{2}\xi^{T}G(x)^{-1}\xi}e^{-U(x) -\frac{1}{2}v^{T}G(x)^{-1}v}(f(x,\xi)-f(x,v))d \xi d\lambda_{TM}(z)\\
    &=  \int C(x)e^{-U(x)}\int \! \int e^{-\frac{1}{2}\xi^{T}G(x)^{-1}\xi}e^{-\frac{1}{2}v^{T}G(x)^{-1}v}(f(x,\xi)-f(x,v))d \xi dvd\sigma_{\mathcal{M}}(x)\\
    &= 0,
\end{align*}
where $C(x)$ is a varying constant depending on $x$.
 Therefore we have that
 \[\int_{T\mathcal{M}}\mathcal{L}f(x,v) \mu(dz) = 0, \]
 and $\mu$ is an infinitesimally invariant measure.
\end{proof}

We will next demonstrate that $C^{1}_{c}(T\mathcal{M})$ is a core of $D(A)$ by showing that certain conditions established in \cite{Du2021} hold under the assumption that $\mathcal{M}$ is compact. To show that $C^{1}_{c}(T\mathcal{M})$ is a core of $D(A)$ we use the approach of compactly approximating RT-RMHMC by a more well-behaved PDMP, which has PDMP characteristics $(\varphi,\lambda,Q^{\epsilon})$ satisfying the Assumption A3 from \cite{Du2021} and has a Feller transition semigroup $(P_{t})_{t \geq 0}$. We then use this approximation to show that RT-RMHMC is Feller and $C^{1}_{c}(T\mathcal{M})$ is a core of the strong generator of RT-RMHMC, whose transition semigroup $(P_{t})_{t \geq 0}$ is seen as a semigroup on $C_{0}(T\mathcal{M})$. Note that $C_{0}(T\mathcal{M})$ denotes the space of continuous functions $f : T\mathcal{M} \to \mathbb{R}$ that vanish at infinity and $C_{0}(T\mathcal{M})$ is a Banach space when equipped with the $\lVert\cdot \Vert_{\infty}$ norm.

We first approximate our PDMP $(\varphi,\lambda,Q)$ (RT-RMHMC) with the PDMP with characteristics $(\varphi,\lambda, Q^{\epsilon})$ in the sense that
\[ 
\sup_{z \in T\mathcal{M}, A \in \mathcal{B}(T\mathcal{M})} \{\lambda^{\epsilon}(z) \wedge \lambda(z) \lvert Q^{\epsilon}(z,A) - Q(z,A)\vert + \lvert \lambda^{\epsilon}(z) - \lambda(z) \vert\} \leq \epsilon,
\]
where $Q^{\epsilon}$ is constructed as a Markov kernel corresponding to a consistently truncated Gaussian distribution on each tangent space as follows.

Define $G: \mathcal{M} \times \mathbb{R} \to \mathbb{R}$ by 
\[G(x,a) = \int_{B(0,a)} \psi(x)(dv) - (1-\nicefrac{\epsilon}{2\lambda}),\]
where $\psi(x)$ denotes the probability density function of the Gaussian distribution on $T_{x}\mathcal{M}$ defined by $\psi(x)(dv) \propto \exp{(-\frac{1}{2}v^{T}G(x)^{-1}v)}dv$, known as the Maxwellian distribution. Then we have that $\nicefrac{\partial G}{\partial a} \neq 0$ due to the fact that $G(x, \cdot)$ is strictly increasing. By the implicit function theorem there exists a unique continuously differentiable function $M: \mathcal{M} \to \mathbb{R}$ such that $G(x,M(x)) = 0$ for all $x \in \mathcal{M}$.
We define the transition kernel  as follows:

\[ Q^{\epsilon}(z,dz') = \lambda \delta_{x}(dx')\psi_{\epsilon}(x)(dv'),\]
where
\[\psi_{\epsilon}(x)(dv') = \begin{cases}
\frac{1}{1-\nicefrac{\epsilon}{2\lambda}}\psi(x)(dv') & \text{for } \lvert v'\vert_{g} \leq M(x)\\
0 & \text{otherwise}
\end{cases} \]
is the truncated Maxwellian distribution. Then we have that for any $(x,v) \in T\mathcal{M}$ and $A \in \mathcal{B}(T\mathcal{M})$
\begin{align*}
    &\lvert Q^{\epsilon}((x,v),A) - Q((x,v),A)\vert =\\ &=\lambda\lvert(\frac{1}{1-\nicefrac{\epsilon}{2\lambda}}-1)\int_{A \cap B(0,M(x))}\psi(x)(dv') - \int_{A \cap B(0,M(x))^{c}}\psi(x)(dv')\vert \\
    &\leq \lambda(\frac{1}{1-\nicefrac{\epsilon}{2\lambda}}-1)\lvert\int_{A \cap B(0,M(x))}\psi(x)(dv')\vert + \lambda\lvert\int_{A \cap B(0,M(x))^{c}}\psi(x)(dv')\vert\\
    &\leq \lambda(\frac{1}{1-\nicefrac{\epsilon}{2\lambda}}-1)(1-\nicefrac{\epsilon}{2\lambda}) + \lambda(1 - (1-\nicefrac{\epsilon}{2\lambda}))\\
    &=\epsilon.
\end{align*}

\begin{lemma}[Continuity of Semigroup]\label{Continuity of Semigroup}
Let $(\mathcal{M},g)$ be a smooth Riemannian manifold, and let $U \in C^{1}(\mathcal{M})$ and hence $H \in C^{1}(T\mathcal{M})$. Let $(P_{t})_{t \geq 0}$ be the transition semigroup of $(\varphi,\lambda,Q^{\epsilon})$, then 
\[ 
\lvert P_t f - f \vert  \to 0 \text{ as } t \to 0, \text{ for all } f \in C_{0}(T \mathcal{M}).
\]
\end{lemma}

\begin{proof}
Let $(Z_{t})_{t \geq 0}$ denote a sample path of $(\varphi,\lambda,Q^{\epsilon})$. We have that
\begin{align*}
    \lvert P_{t}f(z) - f(z)\vert &= \lvert\mathbb{E}_{z}(f(Z_{t})-f(z)\vert\\
    &= \lvert\mathbb{E}_{z}(f(Z_{t})\mathbbm{1}(S_{1} \leq t) + f(Z_{t})\mathbbm{1} (S_{1} > t))-f(z)\vert\\
    &\leq \lVert f\Vert_{\infty}\mathbb{P}(S_{1} \leq t) + \lvert f(\varphi_{t}(z))e^{-\lambda t}-f(z)\vert\\
    &=\lVert f\Vert_{\infty}(1-e^{-\lambda t}) + \lvert f(\varphi_{t}(z)-f(z)\vert \to 0,
\end{align*}
where $S_{1}$ is the time of the first event and $\varphi_{t}(z)$ is the solution of the Hamiltonian flow. If $H$ is continuously differentiable everywhere then $\varphi_{t}(z)$ is well defined for all $t>0$, and $\varphi_{t}(z) \to z$ as $t \to 0$ (see for example \cite{Ch2006}[Theorem 1.186]).
\end{proof}

\begin{lemma} \label{compactness-result}
Let $(\mathcal{M},g)$ be a compact Riemannian manifold, and let $U \in C^{1}(\mathcal{M})$. Let $(\varphi,\lambda,Q^{\epsilon})$ be the PDMP approximation of RT-RMHMC defined above. The set of all possible sample paths of $(\varphi,\lambda,Q^{\epsilon})$  with initial condition $(X_{0},V_{0})$ is contained in a compact set.
\end{lemma}
\begin{proof}
Let $M(x)$ denote the continuous function in the definition of $Q_{\epsilon}$ which controls the truncation of the Gaussian distribution. $M(x)$ is a continuous function on a compact set and hence bounded by $M_{\epsilon}$. We further choose $M_{\epsilon}$ such that $\lvert V_{0}\vert_{g} \leq M_{\epsilon}$. Define the set
\[ U_{\epsilon} =  \{(x,v) \mid x \in \mathcal{M}, \lvert v\vert_{g} \leq M_{\epsilon}\} \subset T\mathcal{M}.\]
Due to the fact that $\mathcal{M}$ is compact it follows that $U_{\epsilon}$ is a compact subset of $T\mathcal{M}$ by Lemma \ref{compactness theorem}.  We have that $H$ restricted to $U_{\epsilon}$ is bounded by $M_{H}$ as it's continuous on a compact set. We also have that $H$ is constant between event times of the PDMP, by the definition of Hamiltonian flow. Therefore the Hamiltonian defined on the PDMP $(X_{t},V_{t})$ takes values which are defined by the image of $(X_{t_{i}},V_{t_{i}})$, for events $t_{i}$  $i = 1,2,...$. At event time $t_{i} \sim \exp{\lambda}$, we have that $(X_{t_{i}},V_{t_{i}}) \in U_{\epsilon}$, where $(X_{t_{i}},V_{t_{i}}) \sim Q(X_{t_{i}-},V_{t_{i}-}, \cdot)$. Therefore we can bound the Hamiltonian by $M_{H}$ on $\{(X_{t},V_{t})\mid t \geq 0 \}$. Now we have that 
\begin{align*}
    M_{H} \geq H(X_{t},V_{t}) &= U(X_{t}) + \lvert V_{t}\vert_{g} \\
    &\geq m_{U} + \lvert V_{t}\vert_{g}
\end{align*}
for all $t$. Therefore 
\[
\{(X_{t},V_{t}) \mid t \geq 0 \} \subset \{(x,v) \mid x \in \mathcal{M}, \lvert v\vert_{g} \leq M_{V} := M_{H} - m_{U} \},
\]
which is compact by Lemma \ref{compactness theorem}.
\end{proof}
We have the following assumption from \cite{Du2021}, which we use to establish Proposition \ref{Feller-Core-Approx}.
\begin{definition} \cite{Du2021}[Definition 16]
We say that a homogeneous differential flow $\varphi$ on $T\mathcal{M}$ and a homogeneous Markov kernel $Q$ on $T\mathcal{M}$ are compactly compatible if for all compact sets $K \subset T\mathcal{M}$ and $T \geq 0$, there exists a compact set $\Tilde{K} \subset T \mathcal{M}$ satisfying: for all $n \in \mathbb{N}^{*}$, $(t_{i})_{i \in \llbracket  1,n \rrbracket} \in \mathbb{R}^{n}_{+}, \sum^{n}_{i=1}t_{i} \leq T,$ there exists a sequence $(K_{i})_{i \in \llbracket 1,n \rrbracket}$ of compact sets of $T\mathcal{M}$ such that, setting $K_{0} = K,$
\begin{enumerate}
    \item for all $i \in \llbracket 1 , n  \rrbracket,$ $K_{i}$ only depends on $(t_{j})_{j \in \llbracket 1,n \rrbracket}$ and $\cup^{n}_{i=0} K_{i} \subset \Tilde{K}$;
    \item for all $i \in \llbracket 0, n-1 \rrbracket$, $s_{i+1} \in [0,t_{i+1}]$ and $s_{n+1} \in [0,T - \sum^{n}_{j=1} t_{j}],$
    \[\bigcup_{x \in K_{i}}\textnormal{supp}\{Q(\varphi_{t_{i+1}}(x),\cdot) \}\subset K_{i+1}, \qquad \varphi_{s_{i+1}}(K_{i}) \subset \Tilde{K}, \qquad \varphi_{s_{n+1}}(K_{n}) \subset \Tilde{K}. \]
\end{enumerate}
\end{definition}

\begin{assumption}\label{assumption-3}\cite{Du2021}[A3] The homogeneous characteristics $(\varphi,\lambda,Q)$ satisfy
\begin{enumerate}
    \item the flow $\varphi$ and the Markov kernel $Q$ are compactly compatible;\\
    \item $\lambda \in C^{1}(T\mathcal{M})$ and for all $f \in C^{1}(T\mathcal{M})$, $\lambda Q^{\epsilon}f \in C^{1}(T\mathcal{M})$ and there exists a locally bounded function $\Psi:T\mathcal{M} \to \mathbb{R}_{+}$ such that for all $x \in K$,
\[\lambda\lVert\nabla(Q^{\epsilon}f)(x) \Vert \leq \lVert\Psi\Vert_{\infty,K}\sup\{\lvert f(y)\vert+\lVert\nabla f(y) \Vert : y \in \textnormal{supp}\{Q^{\epsilon}(x,\cdot) \}\}; \]
    \item $(t,x) \mapsto \varphi_{t}(x) \in C^{1}(\mathbb{R}_{+} \times T\mathcal{M})$ and for all compact  $K \subset T\mathcal{M}$  and  $t \geq 0$,
\[\sup{\{\lVert\nabla \varphi_{s}(x)\Vert \mid s \in [0,t], x \in K \}} < +\infty. \]
\end{enumerate}
\end{assumption}

\begin{proposition}[Feller and Core of Generator]\label{Feller-Core-Approx}
Let $(P_{t})_{t \geq 0}$ be the transition semigroup of $(\varphi,\lambda,Q^{\epsilon})$ on $T\mathcal{M}$, where $(\mathcal{M},g)$ is a compact smooth Riemannian manifold and $\varphi$ is the Hamiltonian flow associated to the Hamiltonian $H \in C^{2}(T\mathcal{M})$. Then, $(P_{t})_{t \geq 0}$ is Feller and $C^{1}_{c}(T\mathcal{M})$ is a core for the strong generator of $(P_{t})_{t \geq 0}$ seen as a semigroup on $C_{0}(T\mathcal{M})$.
\end{proposition}

\begin{proof}
If we prove that $(\varphi,\lambda,Q^{\epsilon})$ satisfies Assumption \ref{assumption-3}, then from \cite{Du2021}[Theorem 17] $(P_{t})_{t \geq 0}$ satisfies the Feller property. Once the Feller property is established by Lemma \ref{Continuity of Semigroup} and due to the fact that $T\mathcal{M}$ is a complete metric space we have by \cite{Bot2013}[Lemma 1.4] strong continuity of $(P_{t})_{t \geq 0}$ and that $(P_{t})_{t \geq 0}$ is Feller. Further to this $C^{1}_{c}(T\mathcal{M})$ is a core for the strong generator of $(P_{t})_{t \geq 0}$ seen as a semigroup on $C_{0}(T\mathcal{M})$ is a consequence of \cite{Du2021}[Theorem 17] and \cite{Et1986}[Proposition 3.3,Chapter 1].
We will now establish Assumption \ref{assumption-3}.

For any compact set $K \subset T\mathcal{M}$, as $K$ is compact,  $\lvert v\vert_{g} \leq M_{K}$ for some constant $M_{K} \geq 0$ and for all $v$ such that $(\cdot,v) \in K$. Then by the same argument to that of Lemma \ref{compactness-result}, but choosing $M_{\epsilon}$ larger than $M_{K}$ we have that all PDMPs starting in $K$ are contained in a compact set $\Tilde{K}$. We can define $K_{0} = K$ and $K_{i} = \Tilde{K}$ for all $i \geq 1.$ Then we have the flow $\varphi$ and $Q_{\epsilon}$ are compactly compatible and hence Assumption \ref{assumption-3}i) holds.

We show Assumption \ref{assumption-3}ii) as follows. Trivially we have $\lambda \in C^{1}(T\mathcal{M})$. We have taken the metric to be smooth and hence, as the truncated Gaussian distribution has a smooth transition kernel, we have that $Q^{\epsilon}f \in C^{1}(T\mathcal{M})$.   Firstly we note that 
\[
\textnormal{supp}\{Q^{\epsilon}(x,\cdot)\} = \{(x,v) \mid \lvert v\vert_{g} \leq M(x)\} ,
\]
which is compact by Lemma \ref{compactness theorem}.
For all continuously differentiable functions $f: T\mathcal{M} \to \mathbb{R}$, with $(x,y) \in T\mathcal{M}$, we define
\[A(x,y) := \lambda Q^{\epsilon}f(x,y) =  \frac{\lambda}{1-\epsilon/2\lambda} \int_{B(0,M(x))}f(x,y') \psi(x) (dy'). \]
Therefore it is sufficient to show that for all compact sets $K \subset T \mathcal{M}$, and for all $(x,y) \in K$,
\[
\lVert\nabla A(x,y)\Vert \leq \sup_{(w,z) \in K}\{\Psi(w,z) \} \sup \{ \lvert f\vert(x,y') + \lVert\nabla f(x,y')\Vert \mid \lvert y'\vert_{g} \leq M(x)  \},
\]
where $\Psi: T \mathcal{M} \to \mathbb{R}$ is bounded on compact sets of $T \mathcal{M}$. 
Define $\lVert\cdot\Vert_{\infty,M(x)} \equiv \lVert\cdot\Vert_{\infty,B(0,M(x))}$. We have that for all $(x,y) \in T \mathcal{M}$, 
since all functions considered are $C^1$ and hence bounded on all compact sets of $T \mathcal{M}$ we have the following computation which uses the dominated convergence theorem, a Leibniz's integral rule and a spherical coordinate system:
\begin{align*}
\lVert\nabla A(x,y)\Vert &= \frac{\lambda}{1-\epsilon/2\lambda}\lVert\nabla_{x} \int_{B(0,M(x))} f(x,y')\psi(x,y')dy'\Vert\\
&= C\lVert\nabla_{x}\int^{M(x)}_{0} \int_{S^{n-1}}f(x,r,\sigma) \psi(x,r,\sigma) r^{n-1}d\sigma dr\Vert\\
&=C\lVert\int^{M(x)}_{0}\int_{S^{n-1}}\nabla_{x}(f(x,r,\sigma)\psi(x,r,\sigma)r^{n-1})d\sigma dr +\\
& \qquad \qquad \qquad \int_{S^{n-1}}f(x,M(x),\sigma)\psi(x,M(x),\sigma)M(x)^{n-1}d\sigma \cdot \nabla_{x}M(x)\Vert\\
&\leq C \lVert\int_{B(0,M(x))}\nabla_{x}(f(x,y')\psi(x,y'))dy'\Vert +\\
& \qquad \quad C\lVert\int_{S^{n-1}}f(x,M(x),\sigma)\psi(x,M(x),\sigma)d\sigma\Vert \cdot\lVert M(x)^{n-1}\nabla_{x}M(x)\Vert_{\infty}\\
&\leq C_1 \lVert\int_{B(0,M(x))}\nabla_{x}(f(x,y'))\psi(x,y')dy'\Vert + \\
& \qquad \qquad \qquad C_{1}\lVert\int_{B(0,M(x))}f(x,y')\nabla_{x}(\psi(x,y'))dy'\Vert + C_{2}\lVert f(x,\cdot)\Vert_{\infty,M(x)}\\
&\leq C_{1}\lVert\nabla_{x}f(x,\cdot)\Vert_{\infty,M(x)} + C_{2}\lVert f(x,\cdot)\Vert_{\infty,M(x)} + \\
&\qquad \qquad \qquad \qquad C_3 \lVert f(x,\cdot)\Vert_{\infty,M(x)} \lVert\int_{B(0,M(x))}\nabla_{x}\psi(x,y')dy'\Vert\\
& \leq C_{1}\lVert\nabla_{x}f(x,\cdot)\Vert_{\infty,M(x)} + C_{2}\lVert f(x,\cdot)\Vert_{\infty,M(x)} + C_3 \lVert f(x,\cdot)\Vert_{\infty,M(x)} (M(x)^{n}) \\
& \leq C_{1}\lVert\nabla f(x,\cdot)\Vert_{\infty,M(x)} + C_2 \lVert f(x,\cdot)\Vert_{\infty,M(x)}\\
&\leq (C_{1} + C_{2})\lVert\hspace{2mm}\lVert\nabla f(x,\cdot)\Vert + \lvert f(x,\cdot)\vert\hspace{2mm}\Vert_{\infty,M(x)},
\end{align*}
where $C,C_1,C_2$ and $C_3$ are general constants carrying line by line and $\nabla_{x}$ denotes the differential operator with respect to position on $\mathcal{M}$ and we have bounded $\nabla_{x}\psi$ universally on $\{ (x,y) \mid x \in \mathcal{M}, \lvert y\vert_{g} \leq M(x) \} \subset T \mathcal{M}$. Therefore we have the required result by setting $\Psi = C_{1} + C_{2}$. 
Finally we have to show Assumption \ref{assumption-3}iii), where we use the fact that $\varphi$ is continuously differentiable, when $U \in C^{2}(\mathcal{M})$ and for any compact set $K \subset T \mathcal{M}$ we have that $\lvert v\vert_{g} \leq M_{K}$ for all $(x,v) \in T\mathcal{M}$. Then by the same argument as that of Lemma \ref{compactness-result} we can define a larger constant such that all PDMPs starting in $K$ have bounded velocity and hence are contained in a compact set $\Tilde{K}$. Hence Assumption \ref{assumption-3}iii) holds by the fact that a continuous function on a compact set is bounded.
\end{proof}

\begin{theorem}[RT-RMHMC Feller and Core]
Let $(P_{t})_{t \geq 0}$ be the transition semigroup of $(\varphi,\lambda,Q)$ on $T\mathcal{M}$, where $(\mathcal{M},g)$ is a compact smooth Riemannian manifold and $\varphi$ is the Hamiltonian flow associated to the Hamiltonian $H \in C^{1}(T\mathcal{M})$. Then, $(P_{t})_{t \geq 0}$ is Feller and $C^{1}_{c}(T\mathcal{M})$ is a core for the strong generator of $(P_{t})_{t \geq 0}$ seen as a semigroup on $C_{0}(T\mathcal{M})$.
\end{theorem}

\begin{proof}
By construction of $Q^{\epsilon}$ we have the property that
\[ 
\sup_{x \in \mathcal{M}, A \in \mathcal{B}(\mathcal{M})} \{\lambda^{\epsilon}(x) \wedge \lambda(x) \lvert Q^{\epsilon}(x,A) - Q(x,A)\vert + \lvert  \lambda^{\epsilon}(x) - \lambda(x) \} \leq \epsilon.
\]
Using \cite{Du2021}[Theorem 11], Proposition \ref{Feller-Core-Approx}, \cite{Du2021}[Theorem 17] and the same argument as \cite{Du2021}[Theorem 21] we have the required result.
\end{proof}

\begin{corollary}[Invariant measure for RT-RMHMC]\label{cor:invmeasureRTRMHMC}
Let $(P_{t})_{t \geq 0}$ be the transition semigroup of $(\varphi,\lambda,Q)$ on $T\mathcal{M}$, where $(\mathcal{M},g)$ is a compact smooth Riemannian manifold and $\varphi$ is the Hamiltonian flow associated to the Hamiltonian $H \in C^{2}(T\mathcal{M})$. Let $\mu$ be the measure on $(T\mathcal{M}, \mathcal{B}(T\mathcal{M}))$ given by
\[ \mu(dz) \propto e^{-H(x,v)}d\lambda_{T \mathcal{M}}(z),\]
where $d\lambda_{T\mathcal{M}}$ is the Liouville measure of $T\mathcal{M}$.
Then $\mu$ is invariant for RT-RMHMC.
\end{corollary}


\section{Proof of invariance and $\mu$-irreducibility for the Met\-ropolized algorithm}
\label{sec:proofMetropolis}

\begin{proof}[Proof of Proposition \ref{mu-invariant}]
Let $P_{1}$ be the Markov kernel corresponding to the first step. It is clear that resampling from the Gaussian measure on the tangent space keeps $\pi_{\mathcal{H}}$ invariant as it is independent and also keeps $\psi(x)$ invariant, and therefore keeps $\mu$ invariant.

Let $P_{2}$ be the Markov Kernel corresponding the second step (the combination of the sampling the time duration, deterministic step by $\Psi$ and the Metropolis-Hastings accept-reject step. Let $L$ be an arbitrary number of RATTLE steps we will check that $\mu$ is reversible with respect to $P_{2}$ and hence also invariant.

$P_{2}$ is reversible with respect to $\mu$ if for every measurable bounded function $f :T \mathcal{M} \times T \mathcal{M} \to \mathbb{R}$
\[\int \! \int f(z_1, z_2) \mu(dz_1 ) K( z_1 , dz_2 ) = \int \! \int f(z_1 , z_2 ) \mu (d z_2 ) K( z_2 , dz_1 ).  \]

For $P_2$ we have that $P_{2}(z_{1},dz_{2} )$ is non-zero if and only if $z_{2} = \Psi (z_{1} )$ and $z_{2} = z_{1}$, hence we have that
\begin{align*}
    \int \! \int f(z_1 , z_2 ) \mu(dz_1 ) P_2 (z_1 ,dz_2 ) &= \int \! \int f(z_1 , \Psi(z_1)) \min{[1,\exp{{(H(z_1 ) - H(\Psi(z_1 ))}}]}\mu(dz_1) +\\
    &+ \int \! \int f(z_1 , z_1)(1 - \min{[1,\exp{{(H(z_1 ) - H(\Psi(z_1 ))}}]})\mu(dz_1 ).
\end{align*}
Now let $z_2 = \Psi(z_1),$ then due to the momentum reversal map $N$, we have that $z_{1} = \Psi(z_{2}) = \Psi(\Psi(z_1 )),$ and by the volume preserving property of $\Psi$ (preserving the Liouville measure), we have that
\[\mu(dz_{2}) = \mu(dz_{1}) \cdot \frac{\exp{(-H(z_{2}))}}{\exp{(-H(z_{1}))}} = \mu(d z_1 ) \cdot \exp{(H(z_{1}) - H(z_{2}))},\]
and using this property we have that the first part of the above sum can be written as 
\begin{align*}
  &\int \! \int f(z_{1}, \Psi(z_{1})) \min{[1,\exp{(H(z_1 ) - H(\Psi (z_{1}) ))}]}\mu(dz_{1})\\
  &= \int \! \int f(\Psi(z_{2}),z_{2}) \min{[1, \exp{(H(\Psi(z_{2})) - H(z_{2}))}]} \cdot \exp{(H(z_{2})-H(\Psi(z_{2})))} \mu(dz_{2})\\
  &= \int \! \int f(\Psi(z_{2}),z_{2}) \min{[1, \exp{(H(z_{2}) - H(\Psi(z_{2})))}]} \mu(dz_{2}).
\end{align*}
Now considering the second part of the sum, through a change of variables and combining these two equations we have the required result. We therefore have that $\mu$ is reversible with respect to $P_{2}$ and by the same argument and considering $f$ to be the identity we have invariance $P_{2}$ with respect to $\mu$. Due to the fact that this calculation was independent of time we have that $\mu$ is invariant with respect to the Markov kernel of this algorithm.
\end{proof}

\begin{proof}[Proof of Theorem \ref{irreducible}]
Based on \cite{Br2012}[Theorem 3]. Fix $\Delta t > 0$ sufficiently small such that our assumption holds. For a measurable set $A \subset \mathcal{M}$, we can say $A$ is contained in a compact set $K$, which can be covered by $\{B_{r/2}(x) \mid x \in K \}$. Then we have that for some $x' \in K$, $B_{r/2}(x') \cap A$ has positive measure. We can connect $x$ and $x'$ by a sequence of points $x_{0},...,x_{i},...,x_{n}$ for $0 \leq i \leq n$, defined on the geodesic between $x_{0} = x$ and $x_{n} = x'$ such that $d(x,x_{n}) \leq r/2$. We can find unique $v_{0},...,v_{n}$ such that $(x_{i+1},v_{i+1}) = \Psi^{L}_{\Delta t}(x_{i},v_{i})$ by Theorem \ref{accessibility}. We have that \[ K(x_{i},x_{i+1}) > 0\]
due to the Theorem \ref{accessibility} and the fact that $\phi(x_{i})(v_{i}) > 0 $.  Considering the final step we have due to the triangle inequality $\lvert x_{n-1} - \Tilde{x}\vert < r$ for all $\Tilde{x} \in B_{r/2}(x') \cap A$. Hence by the same reasoning and Theorem \ref{accessibility} we have that $K(x_{n-1}, \Tilde{x}) > 0$ for all $\Tilde{x} \in B_{r/2}(x') \cap A$.  Using the fact that $K(x_{i},x_{i+1}) > 0$ for all $0 \leq i \leq n-2$, and  $K(x_{n-1}, \Tilde{x}) > 0$ for all $\Tilde{x} \in B_{r/2}(x') \cap A$ we have that $K^{n}(x,\Tilde{x}) > 0$ for all $\Tilde{x} \in B_{r/2}(x') \cap A$ and
\begin{align*}
    K^{n}(x,A) \geq K^{n}(x, B_{r/2}(x') \cap A) = \int_{B_{r/2}(x') \cap A} K^{n}(x,y) \sigma_{\mathcal{M}}(dy) > 0.
\end{align*}
\end{proof}

%
%
\section{Additional Results}

\begin{lemma}\label{compactness theorem}
	Let $(\mathcal{M},g)$ be a smooth $k$-dimensional Riemannian manifold, let $K \subset\mathcal{M}$ be compact and let $R \in C^{1}(\mathcal{M})$ such that $R(x)>0$ for all $x \in \mathcal{M}$, then the set
	\[\Tilde{K}_{R} := \{(x,v) \mid x \in K, v \in T_{x}\mathcal{M}, \lvert v\vert_{g} \leq R(x)   \} \]
	is a compact subset of $T\mathcal{M}$.
\end{lemma}
Lemma \ref{compactness theorem} can be shown by showing that the embedding of $\Tilde{K}_{R}$ is closed and bounded. Closure can be established by showing that the limit of convergent sequences is contained in $\Tilde{K}_{R}$, using the derivative of the local parametrisation as defined in \cite{Gu1974}[Page 50].

\section{Integrated Autocorrelation and ESS} \label{IAC}
If the MCMC method converges quickly, we have that the variance $\sigma^{2}(\overline{f})$ (the variance of the estimator) is small. 
From the central limit theorem we know that as $N \to \infty$,
\[
\sqrt{N}(\overline{f} - \langle f \rangle) \sim \mathcal{N}(0,a^{2}) 
\]
and hence
\[
\lim_{N \to \infty} \sigma^{2}(\overline{f}) = \frac{a^{2}}{N}, 
\]
where the quantity $a$ is known as the asymptotic variance. 
We have the following result 
\[
a^{2} = \tau_{f} \sigma^{2}(f), 
\]
where $\sigma^{2}(f)$ is the variance of $f$ under the distribution $\pi$ and is independent of the MCMC scheme used (see  \cite{Ro2004}[Chapter 12] for an in depth study). 

We also have 
\[
\tau_{f} = 1 + 2 \sum^{\infty}_{i=1} \textbf{corr}(f(X^{0}),f(X^{i})) ,
\]
which is known as the integrated autocorrelation (IAC).  If all samples are independent, then $\tau_{f} = 1$. MCMC schemes generate correlated samples, thus $\tau_{f} > 1$. 

The IAC ($\tau_{f}$) is a measure of how dependent the samples are and the closer this value is to $1$, the higher the quality of the MCMC samples produced. Note that we will use $X^{\bullet}$ to denote the random variables in a Markov chain and $X_{\bullet}$ to denote the outputs of an MCMC scheme. In the following numerics we approximate the IAC by a Monte Carlo method, that is we create a finite chain $\{f_{n} \}^{N}_{n=1} = \{f(X_{i}) \}^{N}_{n=1}$ from the MCMC schemes we want to test. 
We estimate 
\[
\textbf{corr}(f(X^{0}),f(X^{i})) \approx \frac{c_{f}(i)}{c_{f}(0)},
\]
where
\[
c_{f}(i) = \frac{1}{N-i} \sum^{N-i}_{n=1} (f_{n} - \mu_{f})(f_{n+i} - \mu_{f}) 
\]
and
\[ 
\mu_{f} = \frac{1}{N}\sum^{N}_{n=1}f_{n}.
\]
We have that
\[
\tau_{f} \approx 1 + 2 \sum^{M}_{i = 1} \frac{c_{f}(i)}{c_{f}(0)},
\]
for some large $M$ such that $M \ll N$. Note that in practice one uses a fast Fourier transform method to calculate $c_{f}(\cdot)$ as it is much more computationally efficient.

We now define an additional metric of quality of samples known as effective sample size (ESS) which is defined as
\[
N_{\textnormal{eff}} = \frac{N}{\tau_{f}} 
\]
for a sample size of size $N$. This metric is used to say that a sample of size $N$ of an MCMC algorithm has the efficiency of $N_{\textnormal{eff}}$ independent samples for computing the Monte Carlo average of $f$.




\end{appendices}


\bibliography{Cites.bib}


\end{document}